\documentclass[11pt]{article}

\usepackage{amsmath}
\usepackage{amsthm}
\usepackage{amssymb}
\usepackage{algorithm}
\usepackage{color}
\usepackage{graphicx}
\usepackage{grffile}
\usepackage{multirow}
\usepackage{url}
\usepackage{color}
\usepackage{algpseudocode}
\usepackage[T1]{fontenc}
\usepackage{bbm}
\usepackage{comment}
\usepackage{dsfont}
\usepackage{thm-restate}
\usepackage{bm}
\usepackage{tcolorbox}
\usepackage{subcaption}
\usepackage{enumitem}
\usepackage[margin=1in]{geometry}
\usepackage{mathtools}

\newtheorem{theorem}{Theorem}[section]
 
\newtheorem{lemma}[theorem]{Lemma}

\theoremstyle{definition}
\newtheorem{remark}[theorem]{Remark}
\newtheorem{definition}[theorem]{Definition}

\DeclareMathAlphabet{\bbold}{U}{bbold}{m}{n}
\newcommand{\id}{\ensuremath{\bbold{1}}}

\newcommand{\wh}{\widehat}

\newcommand{\eps}{\epsilon}
\newcommand{\N}{\mathcal{N}}

\newcommand{\R}{\mathbb{R}}

\newcommand{\mL}{\mathcal{L}}
\newcommand{\mH}{\mathcal{H}}
\newcommand{\mZ}{\mathcal{Z}}
\newcommand{\mC}{\mathcal{C}}
\newcommand{\mS}{\mathcal{S}}
\newcommand{\mT}{\mathcal{T}}
\newcommand{\mV}{\mathcal{V}}

\newcommand{\mR}{\mathcal{R}}

\renewcommand{\varepsilon}{\epsilon}
\renewcommand{\hat}{\wh}
\renewcommand{\bar}{\overline}
\renewcommand{\eps}{\epsilon}

\newcommand{\TV}{\mathsf{TV}}
\newcommand{\KL}{\mathsf{KL}}
\newcommand{\chance}{\mathsf{chance}}
\newcommand{\D}{\mathcal{D}}

\DeclareMathOperator*{\E}{\mathbb{E}}

\DeclareMathOperator{\ALG}{ALG}

\DeclareMathOperator{\poly}{poly}
\DeclareMathOperator*{\polylog}{polylog}

\DeclareMathOperator{\unif}{unif}

\title{Fast swap regret minimization and applications to approximate correlated equilibria}
\author{  
	Binghui Peng\footnote{Supported by NSF CCF-1703925, IIS-1838154, CCF-2106429, CCF-2107187, CCF-1763970, AF2212233, COLL2134095, COLL2212745} 
	\\ Columbia University \\  \texttt{bp2601@columbia.edu} 
	\and Aviad Rubinstein\footnote{Supported by NSF CCF-1954927, and a David and Lucile Packard Fellowship} 
	\\ Stanford University \\ \texttt{aviad@stanford.edu}
}
\date{}

\begin{document}
	\maketitle
	
	\begin{abstract}
		We give a simple and computationally efficient algorithm that, for any constant $\varepsilon>0$, obtains $\varepsilon T$-swap regret within only $T = \polylog(n)$ rounds; this is an exponential improvement compared to the super-linear number of rounds required by the state-of-the-art algorithm, and resolves the main open problem of~\cite{blum2007external}. Our algorithm has an exponential dependence on $\varepsilon$, but we prove a new, matching lower bound. 
		
		Our algorithm for swap regret implies faster convergence to $\varepsilon$-Correlated Equilibrium ($\varepsilon$-CE) in several regimes: For normal form two-player games with $n$ actions, it implies the first uncoupled dynamics that converges to the set of $\varepsilon$-CE in polylogarithmic rounds; a $\polylog(n)$-bit communication protocol for $\varepsilon$-CE in two-player games (resolving an open problem mentioned by~\cite{babichenko2017communication,ganor2018communication,goos2018near}); and an $\tilde{O}(n)$-query algorithm for $\varepsilon$-CE (resolving an open problem of~\cite{babichenko2020informational} and obtaining the first separation between $\varepsilon$-CE and $\varepsilon$-Nash equilibrium in the query complexity model).
		
		For extensive-form games, our algorithm implies a PTAS for  {\em normal form correlated equilibria}, a solution concept often conjectured to be computationally intractable (e.g.~\cite{von2008extensive, fujii2023bayes}). 
	\end{abstract}
	
	\setcounter{page}{0}
	\thispagestyle{empty}

	\newpage
	\section{Introduction}
\label{sec:intro}

We consider fundamental questions from online learning and game theory. 
%
In online learning, we seek algorithms that perform well in an unknown, dynamically changing environment. 
Specifically, we consider  algorithms that, on each day, select a (possibly mixed) strategy over $n$ available actions, and receive a reward for each chosen action; the rewards are dynamically adjusted by the unknown environment, possibly by an adaptive adversary who observes the history of the algorithm's actions on previous days.
The standard benchmark for this problem is the {\em external regret}, or the difference between the algorithm's cumulative reward and the single best-in-hindsight action; formally,
\begin{align*}
\texttt{external-regret}:=\max_{i^{*}\in [n]}\sum_{t\in [T]} r_{t}(i^{*}) - \sum_{t\in [T]} \langle p_t, r_t\rangle.
\end{align*}
Here, $T$ is the number of days, $r_t$ is the vector of reward for each action in day $t$, and $p_t$ is the algorithm's mixed strategy (or distribution over actions) in day $t$.
One of the most fundamental results in online learning is the existence of efficient algorithms that have vanishing external regret \cite{littlestone1994weighted,kalai2005efficient,arora2012multiplicative}.

While the bound on external regret is very important, it may be less attractive in highly dynamic environments where no single action performs well over the entire lifetime of the algorithm.
Our focus in this work is on {\em swap regret}%
\footnote{Sometimes also {\em internal regret}; see discussion in Appendix~\ref{app:internal-regret} for detailed discussion of terminology in the literature.},
introduced by~\cite{foster1998asymptotic} in the context of calibrated forecasting. In the forecasting game, a weather forecaster has to forecast the probability of rain on each day: a forecast is {\em calibrated}~\cite{dawid1982well} if, across all the days when the forecaster predicted rain probability $\pi$, the empirical proportion of rainy days indeed approaches $\pi$. If, on the other hand, the empirical proportion approaches $\rho \neq \pi$, then forecaster regrets not {\em swapping} $\pi \rightarrow \rho$. More generally, \cite{foster1998asymptotic}'s work extended the notion of regret to account for such swaps, aka 
compare the algorithm's strategy $p = (p_t)$ against all strategies that can be derived from $p$ by applying a swap function $\phi: [n] \rightarrow [n]$ to $p$'s choices. Formally, let $\Phi_n$ be all swap functions that map from $[n]$ to $[n]$; the swap regret measures the maximum gain one could have obtained when using a fixed swap function over its history strategies
\begin{align}\label{eq:swap-def}
\texttt{swap-regret}:=\max_{\phi \in \Phi_n}\sum_{t\in [T]} \sum_{i\in [n]}p_t(i)r_t(\phi(i)) - \sum_{t\in [T]} \langle p_t, r_t\rangle.
\end{align}


There has been extensive work on minimizing swap regret, e.g.~\cite{foster1998asymptotic, foster1999regret,hart2000simple,cesa2006prediction,stoltz2007learning,blum2007external,hart2013simple,chen2020hedging,anagnostides2022near,anagnostides2022uncoupled}. . 
But  all algorithms proposed to date do not guarantee diminishing regret before  a linear number of days ($T = \Omega(n)$)%
\footnote{In fact, to the best of our knowledge all algorithms proposed to date require a slightly super-linear $T = \Omega(n \log(n))$ number of days.}. For example,~\cite{cesa2006prediction} describe a reduction from external regret to swap regret by considering $n^n$ experts corresponding to each of the $n^n$ possible swap functions.
However, the exponential number of experts/swap functions implies that while simple algorithms can achieve $\varepsilon$-external regret in $\Theta(\log(n))$ days (for arbitrarily small constant $\varepsilon>0$), the algorithm from~\cite{cesa2006prediction}'s reduction requires $\Theta(\log(n^n)) = \tilde{\Theta}(n)$ days, namely exponentially slower. \cite{blum2007external,ito2020tight} show that the $\tilde{\Theta}(n)$ is in fact tight if we restrict the algorithm to pure strategies $p_t \in [n]$. \cite{blum2007external} asked whether the swap regret can be minimized in sublinear time using mixed strategies; to the best of our knowledge, despite its importance (see also applications to game theory below), no progress was made on this question.

Our main result resolves ``the key open problem'' from~\cite{blum2007external}, giving a simple algorithm that achieves $\varepsilon$-swap regret in exponentially faster.
\begin{restatable}[Swap regret minimization]{theorem}{multiMWU}
\label{thm:multi-scale}
Let $n \geq 1$ be the number of actions. 
For any $\eps > 0$, there is an algorithm that obtains at most $\eps$-swap regret in a sequence of $(\log(n)/\eps)^{O(1/\eps)}$ days. 
\end{restatable}

While our result gives exponential improvement for constant $\varepsilon$, the dependence on $\varepsilon$ is exponential. We complement our algorithm with a matching lower bound.

\begin{restatable}[Lower bound]{theorem}{LB}
\label{thm:lower-oblivious}
Let $n$ be the number of actions, $T$ be the total number of days. There exists an oblivious adversary such that any online learning algorithm must have at least 
\[
\Omega\left(\min\left\{\frac{T}{\log (T)},  \sqrt{n^{1-o(1)}T}\right\}\right)
\]
expected swap-regret over a sequence of $T$ days.
\end{restatable}


\subsection*{Game Theory}
In game theory, instead of a single algorithm we study the dynamics between $m \ge 2$ selfish agents (henceforth ``players''). 
Nash's theorem~\cite{nash1950equilibrium,nash1951non} says that every finite game has a Nash equilibrium where players have no incentive to deviate.
However, it has been observed as early as \cite{robinson1951iterative,brown1951iterative} that even in very simple games, natural dynamics may not converge to a Nash equilibrium (see also e.g.~\cite{hart2003uncoupled,milionis2023impossibility}). 
A line of work from the past couple of decades on the complexity of computing (approximate) Nash equilibrium~\cite{daskalakis2009complexity, chen2009settling, chen2015well, rubinstein2015inapproximability, rubinstein2016settling, babichenko2016query, babichenko2017communication, goos2018near} extends these results by showing that {\em no efficient dynamics} can guarantee convergence to a Nash equilibrium.


Perhaps the most important alternative to Nash's equilibrium is Aumann's {\em correlated equilibrium}~\cite{aumann1974subjectivity} --- a relaxation of Nash equilibrium defined as follows:
Consider a trusted centralized {\em correlation device} that sends each player a recommended action in their action set, drawn from a joint distribution $\D$. We say that $D$ is an {\em $\eps$-correlated equilibrium} if no player can gain $\varepsilon$ (in expectation over $D$) by deviating from the correlating device's recommendations%
\footnote{Some authors only allow the player to deviate on a single recommended action; while the definitions coincide for exact correlated equilibrium, ours is stronger for approximate correlated equilibrium. In particular, as pointed by~\cite{ganor2018communication,babichenko2020informational} if each player mixes uniformly over their actions, we trivially obtain a $1/n$-approximate correlated equilibrium w.r.t.~the weaker notion that only considers deviating on a single recommended action. See also discussion of swap vs internal regret in Appendix~\ref{app:internal-regret}.}. Formally, for every player $i$ with action set $A_i$, and for any swap function $\phi_i: A_i \rightarrow A_i$, we have 
\begin{align*}
 \E_{a \sim \D}[u_i(a_i; a_{-i})] \geq \E_{a\sim \D}[u_i(\phi_i(a_i); a_{-i})] -\eps && \text{($\eps$-Correlated Equilibrium)}
\end{align*}
Fortunately, \cite{papadimitriou2008computing,jiang2015polynomial} give LP-based polynomial time algorithms that allow a centralized planner who knows all the players' payoff functions to compute correlated equilibria.
%
%

But what happens when you take away the omniscient centralized planner? Can natural, {\em uncoupled dynamics}%
\footnote{Formally, uncoupled dynamics require that each player chooses their strategy based on the history of play and their own payoff function, in particular they do not directly have access to other players' payoff functions.} 
between selfish agents converge to correlated equilibria?
It is known if every agent minimizes their own swap regret, the dynamics converge to the set of correlated equilibria~\cite{foster1997calibrated,fudenberg1999conditional, cesa2006prediction, blum2007external}; in particular, previous work implies convergence to $\eps$-approximate correlated equilibria in $\tilde{\Theta}(n)$. 
Plugging in our main result, we obtain exponentially faster convergence to the set of correlated equilibria (see open problems by e.g.~\cite{blum2007external,anagnostides2022near}).

\begin{restatable}[Uncoupled dynamics]{corollary}{uncoupledCE}
\label{cor:uncouple}
    Let $n$ be the number of actions. For any $\eps > 0$, there exists an uncoupled dynamic that converges to the set of $\eps$-approximate correlated equilibria of a multi-player normal-form game in $(\log(n))^{O(1/\eps)}$ iterations. 
\end{restatable}

The complexity of finding an approximate correlated equilibrium has also been studied in the {\em query complexity} model, where the algorithm has to access the agents' utility functions via an oracle, and the {\em communication complexity} model, where each agent knows their own utility function, and their goal is to jointly find an approximate correlated equilibrium. For a 2-player, $n$-action game, the previous state of the art protocols for $\varepsilon$-approximate correlated equilibrium have query complexity $\Theta(n^2)$ (brute-force) or communication complexity $\tilde{\Theta}(n)$ (based on~\cite{blum2007external}'s swap regret minimization). Using our main result we obtain optimal protocols in both models, resolving open problems by~\cite{babichenko2017communication,ganor2018communication,goos2018near,babichenko2020informational}. 
\begin{restatable}[Query complexity]{corollary}{queryCE}
\label{cor:query}
    Let $m$ be the number of players,  $n$ be the number of actions. There exists a randomized query algorithm that obtains an $\eps$-approximate correlated equilibrium using at most $m n (\log(mn))^{O(1/\eps)}$ payoff queries, with success probability $1- 1/(mn)^{\omega(1)}$.  
\end{restatable}
We note that this gives the first separation of query complexity of approximate correlated equilibrium and approximate Nash equilibrium (as even the communication complexity of approximate Nash equilibrium is near quadratic~\cite{goos2018near}).

\begin{restatable}[Communication complexity]{corollary}{communicationCE}
\label{cor:communication}
Let $n$ be the number of actions. For any $\eps > 0$, there exists a randomized communication protocol that obtains an $\eps$-approximate correlated equilibrium in a two-player $n$-action game using $(\log(n))^{O(1/\eps)}$ bits of communication, with success probability $1- 1/n^{\omega(1)}$.
\end{restatable}

We also obtain a faster algorithm (in the standard computational model) for computing $\varepsilon$-approximate correlated equilibrium. 

\begin{restatable}[Computational complexity]{corollary}{computeCE}
\label{cor:computational-CE}
Let $m$ be the number of players, $n$ be the number of actions. For any $\eps > 0$, there exists a randomized algorithm that computes an $\eps$-approximate correlated equilibrium in time $mn (\log(mn))^{O(1/\eps)}$, with success probability at least $1 - 1/(mn)^{\omega(1)}$.
\end{restatable}

Beyond normal-form games, several extensions of correlated equilibria have been considered  for Bayesian games, where players have incomplete information about the state of the world, and more generally for extensive-form games, where they may also make decisions or learn information sequentially.
Normal-form correlated equilibria (NFCE) is arguably the simplest extension of correlated equilibria to Bayesian and extensive-form games:  the correlating device sends each player a single signal at the beginning of the game,  independent of state of nature or the Bayesian types of players.
This form of correlated equilibrium satisfies desirable game theoretic properties~\cite{fujii2023bayes} and only requires a single round of communication (see discussion in~\cite{celli19computing}), but computing it is a ``major open problem''~\cite{farina2023polynomial}. 
Much of the work on other notions of correlated equilibrium for Bayesian and extensive form games is inspired by the conjectured intractability of NFCE, e.g.~\cite{von2008extensive,fujii2023bayes}.

Here, we give a PTAS for finding NFCE. Moreover, our algorithm can be implemented as uncoupled dynamics by distributed players who each run (a variant of) our algorithm for minimizing swap regret.

\begin{restatable}[Extensive-form games]{corollary}{efgCE}
\label{cor:efg}
Let $m$ be the number of players, $n$ be the number of actions at an information set, $\Phi$ be the number of information sets of a player.
Let $\eps > 0$, there is a randomized uncoupled dynamics algorithm that runs in time $\poly(m, n)\cdot(\Phi\log(n))^{O(1/\eps)}$ and returns an $\eps$-approximate NFCE in an EFG, with success probability $1-1/(mn\Phi)^{\omega(1)}$.
\end{restatable}

\subsection{Related work}

\paragraph{Concurrent work} 
Concurrent and independent work by Dagan, Daskalakis, Fishelson, Golowich \cite{dagan2023external} discovered an algorithm very similar to our swap regret algorithm (Algorithm~\ref{algo:multi}), as well as an equivalent lower bound. Interestingly, they observe that in the same algorithm it is possible to replace the MWU sub-routines with any external regret algorithm; this implies {\em existence} of correlated equilibrium in certain infinite-action games, resolving open problems by Daskalakis and Golowich~\cite{daskalakis2022fast} and Assos et al~\cite{assos2023online}.

\paragraph{No-regret learning in games}
The study of no-regret dynamics in games has been a central topic in the literature of algorithmic game theory and computational learning theory. 
When the game is repeatedly played and each player has diminishing external regret, then the empirical distribution is known to converge to the set of coarse correlated equilibria \cite{foster1993randomization,littlestone1994weighted,freund1997decision,cesa1997use,freund1999adaptive}.
In a coarse correlated equilibrium, a player has no incentive to switch to a fixed action, regardless of the recommended action.
In order to approach the set of correlated equilibria, one has to obtain diminishing swap regret, a problem has been extensively studied in the literature \cite{foster1997calibrated,foster1998asymptotic,foster1999regret,hart2000simple,hart2001reinforcement,cesa2003potential,stoltz2005internal,blum2007external,stoltz2007learning,hart2013simple}.
In particular, the work of \cite{blum2007external} provides an black box reduction from swap regret to external regret, and gives an algorithm that has $O(\sqrt{n\log(n)/T})$ swap regret.
This bound is known to be optimal when the algorithm faces an {\em adaptive adversary} and {\em commits an action} at each round, a matching lower bound is given at \cite{blum2007external, ito2020tight}.
The major open question left by \cite{blum2007external} is whether there exists a faster algorithm that commits a distribution instead an action. We resolve this question.
We refer readers to the book \cite{nisan2007algorithmic,cesa2006prediction} for a general coverage for learning and games.

When all players use the same no-regret learning algorithm, the regret bound can be further improved by exploring the smooth predictable property \cite{daskalakis2011near,rakhlin2013online, rakhlin2013optimization, syrgkanis2015fast,foster2016learning,daskalakis2021near, farina2022kernelized, farina2022near, daskalakis2022fast, chen2020hedging, anagnostides2022near,anagnostides2022uncoupled}. This line of work is initiated by \cite{daskalakis2011near} for zero-sum games and \cite{anagnostides2022near, anagnostides2022uncoupled} provide algorithms obtaining $\tilde{O}(n/T)$ swap regret.
Nevertheless, these algorithms still take $\Omega(n)$ iterations (or even longer) to reach an approximate correlated equilibrium, and it is an open question whether there exists an uncoupled dynamic that leads to correlated equilibria in sublinear or polylogarithmic rounds. See the discussion section of \cite{anagnostides2022near} for a detailed treatment.

\paragraph{No swap regret learning in leader-follower games}
Motivated the attractiveness of online learning algorithms for strategic agents -both in theory and in practice- a recent line of works explores the potential of ``leaders'' who use adaptive strategies to manipulate ``followers'' running online learning algorithms with predictable structure~\cite{braverman2018selling,deng2019prior,deng2019strategizing,camara2020mechanisms,feng2021convergence,mansour2022strategizing,brown2023learning,haghtalab2023calibrated,cai2023selling}. It is known that while followers running naive (``mean-based'') no external regret algorithms are manipulable, followers who have no swap regret are robust to such manipulations~\cite{braverman2018selling,deng2019strategizing,mansour2022strategizing,haghtalab2023calibrated,brown2023learning}.

\paragraph{Query complexity} 
The query complexity of correlated equilibrium has been studied in the literature \cite{hart2018query, babichenko2015query,goldberg2016bounds}. 
The work of \cite{hart2010long, goldberg2016bounds} observes one can simulate the no-swap regret algorithm (e.g. \cite{blum2007external}) in the query model and finds an approximate correlated equilibrium. In particular, one needs $O(mn^2)\cdot \poly(1/\eps)$ queries to find an $\eps$-approximate correlated equilibrium in an $m$-player $n$-action game.
\cite{hart2018query} proves a query lower bound, showing an exponential number of queries are needed in multi-player games if (1) one wants to find an exact correlated equilibrium; or (2) one uses deterministic algorithm.
The query complexity of Nash equilibrium has been studied, and a query lower bound of $2^{\Omega(m)}$ is known for $m$-player binary action games \cite{babichenko2016query,chen2015well,rubinstein2016settling} and $\Omega(n^2)$ for two-player $n$-action games \cite{goos2018near}. It is an open question whether one can separate the query complexity of Nash and correlated equilibrium in two-player games \cite{babichenko2020informational}.

\paragraph{Communication complexity} The work of \cite{hart2010long} initiates the study of communication complexity of correlated equilibrium and propose to use communication as a complexity measure of uncoupled dynamics. \cite{hart2010long} observes one can use $\poly(n)$ bits of communication to simulate the ellipsoid algorithm of \cite{papadimitriou2008computing, jiang2015polynomial} and finds an exact correlated equilibrium.
\cite{ganor2018communication} gives an $\Omega(n)$ communication lower bound for finding an $1/\poly(n)$-approximate correlated equilibrium in two-player games.
The communication complexity of Nash equilibrium is well studied \cite{babichenko2017communication,goos2018near,roughgarden2016communication, ganor2021communication, babichenko2019communication, babichenko2020communication}. For $m$-player binary action games, the seminal work of \cite{babichenko2017communication} gives a communication lower bound of $\Omega(2^m)$ for finding $\eps$-approximate NE for some constant $\eps > 0$; for two-player $n$-action games, \cite{goos2018near} gives an $\Omega(n^{2-o(1)})$ communication lower bound for finding $\eps$-approximate NE. The communication complexity of correlated equilibrium is an open question repeatedly mentioned in the literature \cite{goos2018near,ganor2018communication,babichenko2020informational}.

We refer readers for the excellent survey of \cite{babichenko2020informational} for a general coverage on the information bounds (query and communication) of equilibria.

\paragraph{Computation of correlated equilibrium} 
For two-player games, an exact correlated equilibrium can be solved via linear programming \cite{hart1989existence}.
For multi-player succinct games, the linear program has exponential size but a correlated equilibrium can be found via ellipsoid methods \cite{papadimitriou2008computing, jiang2015polynomial}. 
The linear programming approach could find the exact (or high accuracy) equilibrium but the runtime is a large polynomial.
The algorithm of \cite{blum2007external} can be used to find an $\eps$-approximate correlated equilibrium in $\Theta(n^3) \cdot \poly(1/\eps)$ time, the qubic barrier comes from solving a linear system ($n^2$) for a total of $n$ iterations.

\paragraph{Extensive-form game and Bayesian games}
The Bayesian game extends the normal-form game by incorporating incomplete information. It is PPAD-hard even to find a constant approximate Bayesian Nash equilibrium in two-player games with $O(1)$ actions \cite{rubinstein2015inapproximability}. 
For correlated equilibria, there are different legitimate definitions for Bayesian games \cite{forges1993five}, see \cite{fujii2023bayes} for an excellent exposure.
Existing work provides uncoupled dynamics to coarse Bayesian correlated equilibrium \cite{hartline2015no} and communication correlated equilibrium \cite{fujii2023bayes}.
The strategic-form correlated equilibrium considered in this paper, is perhaps the most natural one -- it does not reveal any private information to a mediator, and satisfies strong properties such as strategic representability and incentive compatible with strategies. 
However, this comes at price, it is an open question whether one can efficiently find a strategic-form correlated equilibrium, due to the exponential size of the strategy space \cite{fujii2023bayes}. 
We positively answer this open question for arbitrarily small constant approximation.

The extensive-form games extend Bayesian games by incorporating sequential structure and it can be seen as a tree-like Bayesian game, it has, for example, important applications to games like Poker~\cite{brown2018superhuman,brown2019superhuman,brown2019deep}. 
The normal-form correlated equilibrium shares a similar fate as strategic-form correlated equilibrium; while it is natural and satisfies strong properties, it is unclear beforehand one can efficiently find one.
The extensive-form correlated equilibrium, introduced by \cite{von2008extensive}, circumvents the computation challenge by allowing the mediator to release the signal only when reaching the information sets. It admits polynomial time algorithm \cite{von2008extensive,huang2008computing, zhang2022polynomial} and uncoupled dynamics \cite{farina2022simple}. There is a long line of work on extensive-form correlated equilibrium \cite{zinkevich2007regret,lanctot2009monte, farina2019regret, farina2019efficient, farina2019optimistic, farina2022kernelized, zhang2022polynomial, zhang2022optimal, bai2022efficient, anagnostides2022faster,chhablani2023multiplicative,anagnostides2023near} and we refer interested readers to the recent work \cite{farina2023polynomial} for a general coverage.
In particular, our work provides efficient uncoupled dynamics to approximate normal-formed correlated equilibrium, which captures the most rational types of deviation, a major open question in the field, see \cite{farina2023polynomial} for a discussion.

	\section{Preliminary}
\label{sec:pre}
\paragraph{Notation} Let $[n] = \{1,2,\ldots, n\}$ and $[n_1: n_2] = \{n_1, n_1+1, \ldots, n_2\}$. Let $\Delta_n$ be all probability distributions over $[n]$, $1_{n}$ be the uniform distribution over $[n]$, $e_i$ ($i\in [n]$) be the one-hot vector that is $1$ on the $i$-th coordinate and $0$ elsewhere.
Given a vector $r\in \R^{n}$, we use $r(i)$ to denote its $i$-th entry and $\|r\|_{\infty}:= \max_{i \in [n]}|r(i)|$. 
We use $\langle p, r\rangle$ to denote the inner product of two vectors $p, r$. For any $\mu \in [0,1]$, let $B_\mu$ be the  Bernoulli distribution with mean $\mu$.

\subsection{Online learning}
We consider the standard adversarial online learning setting. 
Let $T$ be the total number of days, $n$ be the number of experts and $B > 0$ be the width of reward sequence.
There is a sequence of $T$ days and at each day $t \in [T]$, the algorithm plays a distribution $p_t \in \Delta_n$ over the set of action $[n]$. After that, the adversary selects a reward vector $r_t \in [0,B]^{n}$. The algorithm observes $r_t$ and receives reward $\langle p_t, r_t\rangle$.
At the end of sequence, the {\em external regret} measures the maximum gain one would have achieved when switching to a fixed action
\begin{align*}
\texttt{external-regret}:=\max_{i^{*}\in [n]}\sum_{t\in [T]} r_{t}(i^{*}) - \sum_{t\in [T]} \langle p_t, r_t\rangle.
\end{align*}
Let $\Phi_n$ be all swap functions that map from $[n]$ to $[n]$, the {\em swap regret} measures the maximum gain one could have obtained when using a fixed swap function over its history strategies
\begin{align*}
\texttt{swap-regret}:=\max_{\phi \in \Phi_n}\sum_{t\in [T]} \sum_{i\in [n]}p_t(i)r_t(\phi(i)) - \sum_{t\in [T]} \langle p_t, r_t\rangle.
\end{align*}

\begin{remark}[Model of adversary]
In the literature of online learning, an oblivious adversary (randomly) chooses the reward vector $r_1, \ldots, r_{T}$ at the beginning.
An adaptive adversary could choose the reward vector $r_t$ based on the algorithm's history strategy $p_1, \ldots, p_{t-1}$. 
A strong adaptive adversary could further observe the strategy $p_t$ of the current round.
Our algorithm holds against the strong adaptive adversary while our lower bound rules out better algorithms against oblivious adversary. 
We note that the adaptive adversary model is sufficient for applications on correlated equilibria.
\end{remark}

\subsection{Correlated equilibria and swap regret}
The most important application of swap regret minimization is its connection with the {\em correlated equilibrium} in game theory. In an $m$-player normal-form game, each player $i\in [m]$ has an action set $A_i$ ($|A_i| = n$). Given an action profile $(a_1, \ldots, a_m) \in A_1 \times \cdots \times A_m$, the $i$-th player receives utility $u_i(a_i; a_{-i})\in [0,1]$. A correlated equilibrium is a joint distribution over the action space such that no one has the incentive to deviate from its recommended action.

\begin{definition}[$\eps$-correlated equilibrium]
A joint probability distribution $\D$ over $A_1 \times \cdots \times A_m$ is an $\eps$-correlated equilibrium if for every player $i \in [m]$ and for any swap function $\phi_i: A_i \rightarrow A_i$, we have 
\[
\E_{a \sim \D}[u_i(a_i; a_{-i})] \geq \E_{a\sim \D}[u_i(\phi_i(a_i); a_{-i})] -\eps.
\]
\end{definition}


It is well-known that if every player locally runs a no-swap regret learning algorithm, then the empirical distribution converges to a correlated equilibrium. In particular, 
\begin{lemma}[Swap regret and correlated equilibrium \cite{foster1997calibrated, blum2007external}]
\label{lem:uncouple1}
If an $m$-player normal-form game is played repeatedly for $T$ days, and each player incurs no more than $R(T)$ swap regret over the $T$ days, then the empirical distribution of the joint actions by the players is an $R(T)/T$-correlated equilibrium. 
\end{lemma}

\subsection{Useful tools}
We make use of the classic algorithm of Multiplicative Weights Update (MWU). 

\begin{algorithm}[!htbp]
\caption{MWU}
\label{algo:mwu}
\begin{algorithmic}[1] 
\State {\bf Input parameters} $T$ (number of rounds),  $n$ (number of actions), $B$ (bound on payoff) 
\For{$t=1,2, \ldots, T$}
\State Compute $p_t \in \Delta_n$ over experts such that $p_t(i) \propto \exp(\eta \sum_{\tau=1}^{t-1}r_\tau(i))$ for $i\in [n]$
\State Play $p_t$ and observes $r_t \in [0, B]^n$
\EndFor
\end{algorithmic}
\end{algorithm}

MWU has small external regret against a strong adaptive adversary.
\begin{lemma}[\cite{arora2012multiplicative}]
\label{lem:mwu}
Let $n, T \geq 1$ and the reward $r_t \in [0,B]^n$ ($t\in [T]$). 
If one takes $\eta = \sqrt{\log (n)/T}/B$, then the MWU algorithm guarantees an external regret of at most
\begin{align*}
\max_{i^{*}\in [n]}\sum_{t \in [T]}r_t(i^{*}) - \sum_{t \in [T]} \langle p_t, r_t\rangle \leq \frac{\log(n)}{\eta} + \eta T B^2 \leq 2B\sqrt{T\log (n)}
\end{align*}
against a strong adaptive adversary.
\end{lemma}

\section{Multi-scale MWU}

Our goal is to prove

\multiMWU*


Let $S := \log_2(1/\eps) + 1$, and let $H :=  4\log(n)2^{2S} = \Theta(\log(n)/\eps^2)$ be the block size. Algorithm \ref{algo:multi} runs MWU in multiple scales: It maintains $2^{S}$ threads of MWU over a sequence of $T = H^{2^{S}}$ days. 
The $k$-th thread ($k\in [2^{S}]$) restarts every $T/H^{k}$ days, and each restart lasts for $H^{k}$ days. 
During each restart, it views $H^{k-1}$ days as one ``meta day'' and executes MWU for $H$ steps (Line \ref{line:meta-day1} -- \ref{line:meta-day2}). 
The final algorithm aggregates $2^{S}$ threads by playing uniformly over them.

\begin{algorithm}[!htbp]
\caption{Multi-scale MWU}
\label{algo:multi}
\begin{algorithmic}[1]
\State {\bf Input parameters} $T$ (number of rounds),  $n$ (number of actions), $B$ (bound on payoff) 
\State {\bf Internal parameters} $H, S$ such that $T = H^{2^{S}}$ 

\For{$t = 1,2,\ldots, T$}
\State Let $q_{k, t} \in \Delta_{n}$ be the strategy of $\textsc{MWU}_k$ ($k \in [2^{S}]$), play uniformly over them 
\begin{align*}
p_{t} = \frac{1}{2^{S}}\sum_{k\in [2^{S}]}q_{k, t}
\end{align*}
\EndFor
\Procedure{$\textsc{MWU}_k$}{} \Comment{$k \in [2^{S}]$}
\For{$\ell = 1,2,\ldots, T/H^{k}$} \Comment{Restart every $H^{k}$ days}
\State Initiate MWU with parameters $H, n, H^{k-1}B$ \label{line:meta-day1}
\For{$h = 1,2,\ldots, H$}
\State Let $z_{\ell, h}\in\Delta_n$ be the strategy of MWU at the $h$-th round, play $z_{\ell, h}$ for $H^{k-1}$ days
\State Update MWU with the aggregated rewards of the last $H^{k-1}$ days 
\[
\left\{\sum_{\tau =(\ell-1) H^{k}+ (h-1) H^{k-1} + 1}^{(\ell-1) H^{k}+ h H^{k-1}} 
r_{\tau}(i)\right\}_{i \in [n]} \in [0, H^{k-1}B]^n
\]
\EndFor \label{line:meta-day2}
\EndFor
\EndProcedure
\end{algorithmic}
\end{algorithm}

\begin{proof}
Fix the block size $H$, and let $\delta = 2\sqrt{\log(n)/H}$. Let $T_{S} = H^{2^{S}}$, we prove that the total swap regret of Multi-scale MWU (Algorithm \ref{algo:multi}) over a sequence of $T_S$ days is at most 
\begin{align}
 2^{-S}\left(\sum_{t \in [T_{S}]} \|r_{t}\|_{\infty} - \Big\|\sum_{t\in [T_S]}r_{t}\Big\|_{\infty} \right) +  \delta T_{S} B. \label{eq:multi-scale-goal}
\end{align}
We prove Eq.~\eqref{eq:multi-scale-goal} by induction on $S$.
The base case of $S = 0$ holds due to the external regret guarantee of MWU. Concretely, for any swap function $\phi: [n]\rightarrow [n]$, the swap regret satisfies
\begin{align*}
\sum_{t \in [T_{0}]} \sum_{i\in [n]}p_{t}(i)r_{t}(\phi(i)) - \sum_{t \in [T_{0}]} \sum_{i\in [n]}p_{t}(i)r_{t}(i) \leq &~ \sum_{t \in [T_{0}]} \|r_t\|_{\infty} - \Big\|\sum_{t \in [T_{0}]} r_{t}\Big\|_{\infty} + 2\sqrt{\log(n)T_0} B\\
= &~ \left(\sum_{t \in [T_{0}]} \|r_{t}\|_{\infty} - \Big\|\sum_{t \in [T_{0}]} r_{t}\Big\|_{\infty}\right) +  \delta T_0 B.
\end{align*}
where the first step holds due to $r_t(\phi(i)) \leq \|r_t\|_{\infty}$ ($i \in [n]$) and the external regret guarantee of MWU. The second step follows from the definition of $\delta$.

Suppose the claim holds up to $S = s$, we prove that it continues to hold for $S = s+1$. 
We divide $[T_{s+1}]$ into $T_s = H^{2^{s}}$ intervals. For the $\tau$-th ($\tau \in [T_s]$) interval $[(\tau-1)T_{s}+1: \tau T_{s}]$, let $R_{\tau}(i)$ be the total reward of action $i \in [n]$, i.e.,
\[
R_{\tau}(i) := \sum_{t\in [(\tau-1) \cdot T_s + 1: \tau T_{s}]}r_{t}(i) \in [0, T_s B]
\]
For any swap function $\phi: [n]\rightarrow [n]$, we split the regret into two parts, one for threads $[2^{s}]$ and one for threads $[2^{s}+1:2^{s+1}]$
\begin{align}
&~ \sum_{t \in [T_{s+1}]} \sum_{i\in [n]}p_{t}(i)r_{t}(\phi(i)) - \sum_{t \in [T_{s+1}]} \sum_{i\in [n]}p_{t}(i)r_{t}(i)\notag \\
= &~  \frac{1}{2^{s+1}}\sum_{t \in [T_{s+1}]} \sum_{i\in [n]} \sum_{k \in [2^{s+1}]} q_{k, t}(i) (r_{t}(\phi(i)) - r_t(i)) \notag \\
= &~ \frac{1}{2^{s+1}} \sum_{t \in [T_{s+1}]} \sum_{i\in [n]}\sum_{k \in [2^{s}]}  q_{k, t}(i) (r_{t}(\phi(i)) - r_t(i)) \notag \\ 
&~ + \frac{1}{2^{s+1}}\sum_{t \in [T_{s+1}]} \sum_{i\in [n]} \sum_{k \in [2^{s}+1:2^{s+1}]}  q_{k, t}(i) (r_{t}(\phi(i)) - r_t(i)) \label{eq:multi-scale-mwu1}
\end{align}
Here the first step holds since the algorithm plays uniformly over $2^{s+1}$ threads, that is,  $p_{t} = \frac{1}{2^{s+1}}\sum_{k \in [2^{s+1}]}q_{k, t}$. 


We bound each of the two sums in Eq.~\eqref{eq:multi-scale-mwu1} separately. 
For the first $2^{s}$ threads, we have
\begin{align}
\sum_{t \in [T_{s+1}]} \sum_{i\in [n]} \sum_{k \in [2^{s}]}  q_{k, t}(i) (r_{t}(\phi(i)) - r_t(i))
= &~ \sum_{\tau \in [T_s]} \sum_{t\in [(\tau-1)T_s+1: \tau T_{s}]} \sum_{i\in [n]}\sum_{k \in [2^{s}]}  q_{k, t}(i) (r_{t}(\phi(i)) - r_t(i))\notag \\
\leq &~ \sum_{\tau \in [T_{s}]}\left( \left(\sum_{t\in [(\tau-1)T_s+1: \tau T_{s}]}\|r_t\|_{\infty} - \|R_{\tau}\|_{\infty} \right) + 2^{s} \cdot \delta T_{s} B\right)\notag \\
= &~ \left(\sum_{t\in [T_{s+1}]}\|r_t\|_{\infty} - \sum_{\tau\in [T_s]}\|R_\tau\|_{\infty}   \right) + 2^s \cdot \delta T_{s+1} B. \label{eq:multi-scale-mwu2}
\end{align}

In the first step, we split the swap regret into $T_{s}$ intervals. The second step follows from the inductive hypothesis. In particular, for each interval $\tau \in [T_s]$, playing uniformly over threads $[2^{s}]$ is equivalent to running multi-scale MWU for $T_{s}$ days with width $B$.

For each thread $k \in [2^{s}+1: 2^{s+1}]$, the strategy $q_{k, t} \in \Delta_{n}$ is fixed within each interval $\tau \in [T_{s}]$. That is, we can define 
\[
w_{k, \tau} := q_{k, (\tau-1)T_s+1} = \cdots =q_{k, \tau T_s} \quad \forall k \in [2^{s}+1: 2^{s+1}], \tau \in [T_{s}].
\]
Then, we have
\begin{align}
 &~\sum_{t \in [T_{s+1}]} \sum_{i\in [n]} \sum_{k \in [2^{s}+1:2^{s+1}]}  q_{k, t}(i) (r_{t}(\phi(i)) - r_t(i))  \notag \\
 =  &~  \sum_{\tau \in [T_{s}]} \sum_{i\in [n]}\sum_{k \in [2^{s} + 1: 2^{s+1}]}  w_{k, \tau}(i) (R_{t}(\phi(i)) -R_t(i))\notag \\
\leq &~ \left(\sum_{\tau \in [T_{s}]}\|R_\tau\|_{\infty} - \Big\|\sum_{\tau\in [T_s]} R_\tau \Big\|_{\infty}   \right) + 2^{s}\cdot \delta T_{s} \cdot (T_{s} B) \notag \\
= &~ \left(\sum_{\tau \in [T_{s}]}\|R_\tau\|_{\infty} - \Big\|\sum_{t\in [T]}r_t \Big\|_{\infty}   \right) + 2^{s}\cdot \delta T_{s+1} B. \label{eq:multi-scale-mwu3}
\end{align}
The first step follows from the definition of $w_{k, \tau}$ and $R_{\tau}$. 
The second step follows from the inductive hypothesis. In particular, by viewing each interval as one meta day, playing uniformly over threads $[2^{s}+1: 2^{s+1}]$ is equivalent to running multi-scale MWU for $T_{s}$ days with width $T_s B$. The last step follows from the definition of $R_{\tau}$.

Combining Eq.~\eqref{eq:multi-scale-mwu1}\eqref{eq:multi-scale-mwu2}\eqref{eq:multi-scale-mwu3}, we have
\begin{align*}
&~ \sum_{t \in [T_{s+1}]} \sum_{i\in [n]}p_{t}(i)r_{t}(\phi(i)) - \sum_{t \in [T_{s+1}]} \sum_{i\in [n]}p_{t}(i)r_{t}(i)\\
\leq &~ \frac{1}{2^{s+1}}\left(\sum_{t\in [T_{s+1}]}\|r_t\|_{\infty} - \sum_{\tau\in [T_s]}\|R_\tau\|_{\infty}   \right) + \frac{1}{2}\delta T_{s+1} B\\
&~ + \frac{1}{2^{s+1}}\left(\sum_{\tau \in [T_{s}]}\|R_\tau\|_{\infty} - \Big\|\sum_{t\in [T]}r_t \Big\|_{\infty}   \right) + \frac{1}{2}\delta T_{s+1} B\\
= &~ \frac{1}{2^{s+1}}\left(\sum_{t\in [T_{s+1}]}\|r_t\|_{\infty} - \Big\|\sum_{t\in [T]}r_t \Big\|_{\infty} \right) + \delta T_{s+1} B.
\end{align*}
This completes the induction and proves Eq.~\eqref{eq:multi-scale-goal}.

Now, by plugging $S = \log_2(1/\eps) + 1$ and $H = 4\log(n)2^{2S}$ into Eq.~\eqref{eq:multi-scale-goal}, the expected swap regret of multi-scale MWU is at most 
\begin{align*}
\E[\texttt{swap-regret}] \leq  2^{-S}\left(\sum_{t \in [T_{S}]} \|r_{t}\|_{\infty} - \Big\|\sum_{t\in [T_S]}r_{t}\Big\|_{\infty} \right) +  \delta T_{S} B \leq \frac{\eps}{2}\cdot T_{S} B + \frac{\eps}{2}\cdot T_{S} B = \eps  T_{S} B
\end{align*}
in a sequence of 
\[
T_{S} = H^{2^{S}} = (4\log(n)2^{2(\log_2(1/\eps) + 1}))^{2^{\log_2(1/\eps) + 1}} = (16\log(n)/\eps^2)^{2/\eps} =  (\log(n)/\eps)^{O(1/\eps)}
\]
days. 
\end{proof}

	\section{Applications}
\label{sec:application}

The multi-scale MWU obtains diminishing swap regret in the adversarial setting and has many implications for correlated equilibria. 
A direct corollary of Theorem \ref{thm:multi-scale} is the existence of uncoupled dynamics that converge to an approximate correlated equilibrium in polylogarithmic rounds. The proof is a direct combination of Theorem \ref{thm:multi-scale} and Lemma \ref{lem:uncouple1}.

\uncoupledCE*


For most applications appearing in this section, we use the protocol shown at Figure \ref{fig:algo}.
In the protocol, all players repeatedly play the game for $T$ days and each player runs the multi-scale MWU. Instead of calculating the exact reward at every day, each player constructs an approximate estimate of the reward by sampling from other players' mixed strategy.

\begin{figure}[!htbp]
\begin{tcolorbox}[standard jigsaw, opacityback=0,title=Protocol] 
\begin{itemize}
    \item Player $i$ ($i \in [m]$) runs multi-scale MWU (Algorithm~\ref{algo:multi}) for $T$ rounds
    \begin{itemize}
    \item At the $t$-th round ($t\in [T]$), it commits a strategy $p_{i, t} \in \Delta_{n}$ 
    \item It then samples $K = \Theta(\log^2(mn)/\eps^3)$ action profiles $a_{-i,t,1}, \ldots, a_{-i, t, K} \in A_{-i}$ from other players' strategy distribution $p_{-i,t} = \otimes_{i'\in [m]\setminus\{i\}}p_{i', t}$, and constructs the reward vector $\wh{r}_{i, t} \in [0,1]^{n}$:
    \begin{align*}
    \wh{r}_{i, t}(j) = \frac{1}{K}\sum_{k=1}^{K} u_{i}(j; a_{-i, t, k})  \quad \forall j \in A_i.
    \end{align*}
    \end{itemize}
    \item Output the empirical distribution $\frac{1}{T}\sum_{t\in [T]}p_{1, t}\otimes \cdots \otimes p_{m, t}$
\end{itemize}
\end{tcolorbox}
\caption{Protocol}
\label{fig:algo}
\end{figure}

In the rest of this section, we focus on the regime $\eps \leq 1/\log(n)$ -- for smaller approximation $\eps$, the dominant approach is the BM algorithm \cite{blum2007external}.
The following lemma uses the  swap regret guarantee to obtain convergence of the protocol in Figure \ref{fig:algo} to the set of approximate correlated equilibria.

\begin{lemma}
\label{lem:protocol}
Let $m$ be the number of players, $n$ be the number of actions. 
For any $\eps > 0$, suppose each player follows the protocol in Figure \ref{fig:algo} for $T = (\log(n))^{O(1/\eps)}$ days, then with probability at least $1-1/(mn)^{\omega(1)}$, the output is an $\eps$-approximate correlated equilibrium.
\end{lemma}
\begin{proof}
Let $p_t =p_{1,t}\otimes \cdots \otimes p_{m,t}$ be the empirical mixed strategy at day $t \in [T]$.
For any player $i\in [m]$, day $t \in [T]$, let $r_{i, t} \in [0,1]^{n}$ be the expected reward of player $i$, given other players' strategy $p_{-i,t}$, i.e.
\begin{align*}
r_{i, t}(j) = \E_{a_{-i} \sim p_{-i, t}}[u_{i}(j; a_{-i})] \quad \forall j \in A_i.
\end{align*}
By Chernoff bound, for any action $j \in A_j$, we have
\begin{align*}
\Pr\left[|\hat{r}_{i, t}(j) - r_{i, t}(j)| \geq \frac{\eps}{4}\right] \leq 2\exp(-\eps^2 K /32) \leq (mn)^{-\Omega(\log(mn)/\eps)}. 
\end{align*}
Taking a union bound over $j \in [n], t\in [T], i \in [m]$, with probability at least $1- 1/(mn)^{\omega(1)}$, we have 
\begin{align}
\label{eq:query1}
\left|\wh{r}_{i, t}(j) - r_{i, t}(j)\right| \leq \eps/4 \quad \forall i\in [m], t\in [T], j \in [n].
\end{align}

For any player $i \in [m]$, consider any swap function $\phi_i$, we have
\begin{align*}
\E_{a\sim p}[u_i(\phi_i(a_i); a_{-i})] - \E_{a \sim p}[u_i(a_i; a_{-i})] = &~ \frac{1}{T}\sum_{t\in [T]} \E_{a\sim p_t} \left[u_{i}(\phi(a_i); a_{-i}) - u_{i}(a_i; a_{-i}) \right]\\
= &~ \frac{1}{T}\sum_{t\in [T]}\sum_{j\in [n]} p_{i, t}(j) r_{i, t}(\phi_i(j)) - p_{i, t}(j) r_{i, t}(j) \\
\leq &~ \frac{1}{T}\sum_{t\in [T]}\sum_{j\in [n]} p_{i, t}(j) \wh{r}_{i, t}(\phi_i(j)) - p_{i, t}(j) \wh{r}_{i, t}(j) + \eps/2\\
\leq &~ \eps/2 + \eps/2 = \eps.
\end{align*}
The first step follows from the definition of output distribution $p = \frac{1}{T}\sum_{t\in [T]}p_{t}$, the second step follows from the definition of $r_{i, t}$. The third step holds due to the approximation guarantee of $\wh{r}_{i, t}$ (see Eq.~\eqref{eq:query1}), and the last step holds due to the swap regret guarantee of multi-scale MWU (see Theorem \ref{thm:multi-scale}). 
\end{proof}

\subsection{Query complexity of correlated equilibria}
The first application is for finding an approximate correlated equilibrium using nearly linear number of queries.
Here we consider the standard payoff query model: The utility matrices (tensors for multiplayer games) are unknown but the algorithm can query their entries.
\queryCE*
\begin{proof}
By Lemma \ref{lem:protocol}, the protocol in Figure \ref{fig:algo} is guaranteed to output an $\eps$-approximate correlated equilibrium, with probability at least $1- 1/(mn)^{\omega(1)}$. It remains to bound the total number of queries. For each player $i \in [m]$ and each day $t\in [T]$, it needs $K = O(\log^2(mn)/\eps^3)$ queries to construct one entry of the reward vector $\wh{r}_{i, t}$, and therefore, the total number of query needed is $O(mnTK) = m n (\log(mn))^{O(1/\eps)}$. We complete the proof here.
\end{proof}

\subsection{Communication complexity of correlated equilibrium}
The multi-scale MWU algorithm also gives a communication protocol for finding approximate correlated correlated in two-player normal-form game, using only {\em polylogarithmic} number of bits.
Recall in the communication model, each player knows its own utility, but not others' utility. 
The goal is to output an (approximate) correlated equilibrium with small amount of communication. 

\communicationCE*
\begin{proof}
Consider the following communication protocol.
Alice runs the multi-scale MWU for $T = (\log(n))^{O(1/\eps)}$ days.
At day $t \in [T]$, Alice commits a strategy $p_t \in \Delta_n$. 
Alice samples a multi-set of $K = O(\log^2(n)/\eps^3)$ actions $i_{t,1}, \ldots, i_{t, K}$ from $p_t$ and sends it to Bob.
Bob plays the best response $j_t \in [n]$ to the uniform strategy $\unif(\{i_{t, k}\}_{k\in [K]})$ and sends $j_{t}$ to Alice.
Alice constructs the reward vector as $r_{t}(i) = u_{A}(i; j_t)$ for all $i\in [n]$.
The communication protocol proceeds in $T$ rounds, and at the end, Alice reports the empirical distribution $p = \frac{1}{T}\sum_{t\in [T]}p_t \otimes e_{j_t}$.

We first prove the empirical distribution $p$ is an $\eps$-approximate correlated equilibrium. For Alice, its swap regret is at most $\eps$. Hence, for any swap function $\phi_A: [n] \rightarrow [n]$, one has
\begin{align*}
\E_{a\sim p}[u_A(\phi_A(a_A); a_{B})] - \E_{a \sim p}[u_A(a_A; a_{B})] = \frac{1}{T}\sum_{t\in [T]}\sum_{i\in [n]}p_{t}(i)r_t(\phi_A(i))- p_{t}(i)r_t(i) \leq \eps.
\end{align*}
For Bob, let $\wh{p}_t \in \Delta_n$ be the uniform distribution $\unif(\{i_{t, k}\}_{k\in [K]})$. For any action $j\in [n]$, by Chernoff bound, we have
\begin{align}
\Pr\left[\left|\sum_{i\in [n]}\wh{p}_{t}(i) u_{B}(j; i) - \sum_{i\in [n]}p_{t}(i) u_{B}(j; i)\right| \geq \eps/2 \right] \leq 2\exp(-K\eps^2/8) \leq n^{\Omega(-\log(n)/\eps)}.  \label{eq:communication-approx}
\end{align}
We take an union bound over all actions $j\in [n]$ and days $t \in [T]$, and condition on this event. For any swap function $\phi_B:[n]\rightarrow [n]$, one has
\begin{align*}
\E_{a\sim p}[u_B(\phi_B(a_B); a_{A})] - \E_{a \sim p}[u_B(a_B; a_{A})] = &~ \frac{1}{T}\sum_{t\in [T]}\sum_{i\in [n]}p_{t}(i) u_{B}(\phi(j_t); i) - p_{t}(i) u_{B}(j_t; i)\\
\leq &~ \frac{1}{T}\sum_{t\in [T]}\sum_{i\in [n]}\wh{p}_{t}(i) u_{B}(\phi(j_t); i) - \wh{p}_{t}(i) u_{B}(j_t; i) + \eps \\
\leq &~ \eps.
\end{align*}
The first step follows from the definition of the protocol, the second step follows from Eq.~\eqref{eq:communication-approx}, the third step holds since Bob plays the best response for $\wh{p}_t = \unif (\{i_{t, b}\}_{b\in [B]})$.

The communication complexity of the above protocol is $O(T K \log(n)) = \log(n)^{O(1/\eps)}$. 
\end{proof}




The communication protocol of Corollary \ref{cor:communication} only allows Alice to output the correlated equilibrium. If the goal is a sparse approximate correlated equilibrium that both parties can output, then we can use the following sparisification procedure. The proof can be found at Appendix \ref{sec:application-app}.

\begin{lemma}[Sparsification of correlated equilibrium]
\label{lem:sparsification}
Suppose $p \in \Delta_{n\times n}$ is an $\eps$-approximate correlated equilibrium and its column support has size $S$, i.e., $|\{j: \exists i \in [n], p_{i, j} > 0\}| = S$. Then there is a randomized algorithm that outputs an $(\eps+\delta)$-approximate correlated equilibrium $p'$ that has row support size $O(S^2\log(n)/\delta^2)$ and column support size $S$, without looking at the utility matrices of the game, and with success probability at least $1-1/n^{\omega(1)}$.
\end{lemma}

\subsection{Computational complexity of correlated equilibrium}
Our no-swap regret algorithm gives a nearly linear time algorithm for computing an approximate correlated equilibrium. 
Note that this is {\em sublinear} in the size of description of the game (which is roughly $n^m$).
\computeCE*
\begin{proof}
By Lemma \ref{lem:protocol}, the protocol in Figure \ref{fig:algo} is guaranteed to output an $\eps$-approximate correlated equilibrium, with probability at least $1- 1/(mn)^{\omega(1)}$. It remains to bound the computation cost. For each player $i \in [m]$ and each day $t\in [T]$, it needs to draw $K = O(\log^2(mn)/\eps^3)$ action profiles to construct the reward vector $\hat{r}_{i,t}$. The sampling step takes $O(mnK)$ time for each player. Nevertheless, note these samples can be shared across players, so the total cost for sampling remains $O(mnK)$. The construction of reward vector takes $O(nK)$ time per player, and $O(mnK)$ in total. To maintain the multi-scale MWU, the cost per day equals $O(n/\eps)$ since there are $O(1/\eps)$ threads of MWU. Hence, the total computation cost equals $O(mnKT) = mn(\log(mn))^{O(1/\eps)}$. 
\end{proof}

\subsection{Polynomial time approximation scheme for extensive-form game}

We next give an example showing that the multi-scale MWU can be used to derive polynomial time algorithms for finding approximate correlated equilibrium in large action games. In particular, we present the first polynomial time approximation scheme (PTAS) for computing normal-form correlated equilibrium (NFCE, also known as strategic-form correlated equilibrium) of an extensive-form game (EFG). 
The idea is to use the protocol in Figure \ref{fig:algo} and let each player perform multi-scale MWU over its {\em strategy space}. The strategy space has exponential size but we show that it allows efficient computation.

\paragraph{Extensive-form game}
In an $m$-player extensive-form game, there is a directed game tree $\Gamma$. Let $\N$ be all nodes of $\Gamma$ and $\mZ$ be all terminal nodes.
The non-terminal nodes of the game tree are partitioned into decision nodes and chance nodes $\N \backslash \mZ = \N_{1} \cup \cdots \N_{m} \cup \N_{\chance}$. Here $\N_{i}$ ($i \in [m]$) is the set of nodes where player $i$ takes the action and $\N_{\chance}$ are chance nodes.
The function of a chance node is to assign an outcome of a chance event, and each outgoing edge represents one possible outcome of that chance event as well as the probability of the event. At a decision node, the edges represent actions and successor states that result from the player taking those actions.
The decision nodes of $\N_i$ are further partitioned into {\em information sets} $\mH_i$, and for each information set $h \in \mH_i$, let $A_h$ be all actions available to player $i$. 
The action set $A_h$ is the same for all nodes in $h$,  and it is wlog to assume the action sets $\{A_h\}_{h \in \mH_i}$ are disjoint.
For any information set $h \in \mH_i$, let $\sigma_i(h)$ be the sequence of actions taken by player $i$, from the root to $h$ (it does not include the action taken at $h$).
We assume each player has {\em perfect recall}, i.e., the sequence $\sigma_i(h)$ is the same for every node in the information set $h$.
For terminal nodes, player $i$ receives the reward $\gamma_i(z) \in [0, 1]$ at a terminal node $z \in \mZ$.
The set of pure strategies for player $i \in [m]$ is $\mS_i = \prod_{h \in \mH_i} A_{h}$ and the entire strategy space is $\mS = \prod_{i\in [m]}\mS_i$. 
For simplicity, we assume each player has $\Phi$  information sets, and each information set has $n$ actions.

\paragraph{Notation} For any node $\nu_1, \nu_2 \in \N$, we write $\nu_1 \preceq \nu_2$ if $\nu_1$ is a predecessor of $\nu_2$. Given a strategy profile $s \in \mS$, for each node $\nu \in \N$, let $\pi(s; \nu)$ be the probability of visiting node $\nu$ if players use strategy $s$. Let $u_{i}(s; \nu)$ be the expected utility of player $i$ if it visits node $\nu$, i.e., $u_{i}(s; \nu) := \sum_{z \in \mZ, \nu \preceq z}\pi(s; z) \cdot \gamma_i(z)$.
We use $u_i(s)$ to denote the expected utility of player $i$ at the root. Given an information set $h \in \mH_i$, we write $\nu \in h$ if the decision node $\nu$ is in the information set $h$, let $u_{i}(s; h)$ be the total utility of nodes in $h$, i.e., $u_{i}(s; h) := \sum_{\nu \in h}u_{i}(s; \nu)$.

An $\eps$-approximate NFCE of EFG is a distribution $\sigma \in \Delta(\mS)$ over the strategy space, such that no player can gain $\eps$ more utility (in expectation) by deviating from its recommended strategy.
\begin{definition}[$\eps$-approximate NFCE of EFG]
Let $\eps > 0$, $\sigma \in \Delta(\mS)$ is an $\eps$-approximate normal-form correlated equilibrium of an $m$-player extensive-form game, if for any player $i\in [m]$ and any swap function $\phi: \mS_i \rightarrow \mS_i$, 
\begin{align*}
\E_{s \sim \sigma}[u_i(s_i, s_{-i})] \geq \E_{s \sim \sigma}[u_i(\phi(s_i), s_{-i})] -\eps.
\end{align*}
\end{definition}

The key observation is that one can efficiently implement MWU for extensive-form games.
\begin{lemma}[Efficient implementation of MWU for EFGs]
\label{lem:efg-mwu}
Let $T$ be a positive integer and $\eta > 0$ be the step size. Given strategies $s_{-i, 1}, \ldots, s_{-i, T}\in \mS_{-i}$ of players $[m]\setminus \{i\}$, one can sample from the following distribution in polynomial time 
\begin{align}
p(s_i) \propto \exp\left(\eta \sum_{t \in [T]} u_{i}(s_i, s_{-i, t})\right) \quad \forall s \in \mS_i. \label{eq:mwu-distribution}
\end{align}
\end{lemma}

The proof can be found at Appendix \ref{sec:application-app}. Now, we have
\efgCE*
\begin{proof}
We apply the protocol in Figure \ref{fig:algo} to the strategy space $\mS = \mS_1\times \cdots \times \mS_{m}$. By Lemma \ref{lem:protocol}, the empirical distribution converges to an $\eps$-approximate NFCE in $T = (\log(|\mS_i|))^{O(1/\eps)} = (\Phi\log(n))^{O(1/\eps)}$ days. It remains to demonstrate the computational efficiency. This comes from the fact that each player runs multiple threads of MWU in the protocol, and by Lemma \ref{lem:efg-mwu}, MWU can be efficiently implemented for EFGs. 
\end{proof}

	\section{Lower bound}
\label{sec:lower}
We aim to prove the following lower bound on the swap regret.

\LB*


The lower bound construction is in Section \ref{sec:hard-seq} and its analysis is presented in Section \ref{sec:lower-analysis}.

\subsection{Hard sequence} 
\label{sec:hard-seq}

Let $K, L$ and $\Delta \in (0,1/20]$ be the input parameters.

\paragraph{$K$-ary Tree} The hard sequence goes over all actions $[n]$ via a depth-first search over a $K$-ary tree. 
The tree has $L+1$ levels and each internal node has $K$ child nodes. 
The root is at level $L$ and the leaves are at level $0$. 
Let $\mT_\ell = [0:K-1]^{L-\ell}$ be all nodes at level $\ell \in [0:L]$ and $\mT = \cup_{\ell \in [0:L]}\mT_\ell$ be all nodes in the tree.
We write $a = a_L\ldots a_{\ell+1} \in \mT_{\ell}$ to denote the $a$-th node at level $\ell$, where $a_{\ell+1}, \ldots, a_{L} \in [0:K-1]$.
We write $a.k$ to denote the $k$-th ($k\in [0:K-1]$) child node of $a$.


There are $K^{L}$ leaf nodes in total and each leaf node $a \in \mT_{0}$ maps to two actions $2a + 1, 2a+2$. 
Here we slightly abuse notation and also view $a$ as a natural number in base $K$.
The action set $\N_a$ of an internal node $a \in \mT_{\ell}$ is the union of its descendants' actions. It has size $n_{\ell} = 2K^{\ell}$ and satisfies 
\[
\N_a := \left[\sum_{\ell'= L}^{\ell+1}a_{\ell'} n_{\ell'-1} + 1: \sum_{\ell'= L}^{\ell'+1}a_{\ell'} n_{\ell'-1} + n_\ell\right].
\]
Let $n = 2K^{L}$ be the total number of actions, the root node includes the entire action set $[n] = [2K^{L}]$.


\paragraph{Reward sequence} The reward sequence is formally depicted in Algorithm \ref{algo:hard-sequence-oblivious}.
Nature visits all leaf nodes in order, but randomly skips some of them.
The visit is constructed recursively. Nature starts from the root node, and at each internal node $a \in \mT_{\ell}$ ($\ell \in [L]$) it visits, Nature goes through the $K$ child nodes in order. 
After completing the visit of each child node, Nature has some chance (w.p. $q =\frac{1}{2K}$) to skip the rest of $a$'s sub-tree (Line \ref{line:skip}).
When Nature visits a leaf node $a \in \mL_0$, it constructs the reward sequence for the next $H = \frac{1}{400\Delta^2}$ days as follow.
For nodes that have already been passed, the reward is set to $-1$, i.e., $r_{i} = -1$ for $i \in [2a]$ (Line \ref{line:reward-played1} and Line \ref{line:reward-played2}). For actions $2a+1, 2a+2$, one draws reward from $\frac{L}{16(L+1)} + \frac{1}{16(L+1)} B_{1/2 + \Delta}$ and the other draws reward from $\frac{L}{16(L+1)} + \frac{1}{16(L+1)} B_{1/2}$.
For the rest of action $i \in [2a+3: n]$, consider the path from root to leaf $a$, and suppose $i \in \N_{a'}$ for node $a' \in \mT_\ell$ in the path (if there are multiple such nodes, take the lowest one), then the reward is set to $\frac{L-\ell}{16(L+1)}$ (Line \ref{line:reward-level}).

\begin{algorithm}[!htbp]
\caption{$\textsc{HardSeq}(\ell, a)$ \Comment{Level $\ell \in [0:L]$, node $a \in \mT_\ell$}}
\label{algo:hard-sequence-oblivious}
\begin{algorithmic}[1]
\If{$\ell= 0$} \Comment{Leaf node}
\State Sample $i^{*}(a)\sim \{1,2\}$ 
\State Update reward of action $2a+1, 2a+2$
\begin{align*}
r_{2a + i^{*}(a)} \leftarrow &~ \frac{L}{16(L+1)} + \frac{1}{16(L+1)} B_{1/2 + \Delta} \\
r_{2a + 3 - i^{*}(a)} \leftarrow &~ \frac{L}{16(L+1)} + \frac{1}{16(L+1)} B_{1/2} 
\end{align*}
\State Play for $H = \frac{1}{400\Delta^2}$ days \label{line:two-coin2}
\State Update reward $r_{2a+1}\leftarrow -1,  r_{2a+2} \leftarrow -1$ \label{line:reward-played1}
\Else \Comment{Internal node}
\State Update reward $r_i \leftarrow \frac{L-\ell}{16(L+1)}$ for all actions $i \in \N_a$  \label{line:reward-level}
\For{$k = 0,1,\ldots, K-1$} 
\State $\textsc{HardSeq}(\ell-1, a.k)$ \Comment{Visit the $k$-th child node}
\State {\bf with probability} $q = \frac{1}{2K}$ {\bf do} \Comment{Skip rest of the sub-tree}  \label{line:skip}
\State \indent Update reward $r_{i} \leftarrow -1$ for all action $i\in \N_a$ \label{line:reward-played2}
\State \indent {\bf break} 
\EndFor
\EndIf
\end{algorithmic}
\end{algorithm}

\subsection{Analysis}
\label{sec:lower-analysis}

We analyse the expected swap regret under the reward sequence constructed by Algorithm \ref{algo:hard-sequence-oblivious}. Let $T_{\ALG}$ be the total number of days of Algorithm \ref{algo:hard-sequence-oblivious}, our goal is to prove
\begin{lemma}
\label{lem:lower-goal}
Suppose the reward sequence is constructed as in Algorithm \ref{algo:hard-sequence-oblivious}, then any algorithm has expected swap regret at least
\begin{align}
\label{eq:lower-goal}
\E[\texttt{swap-regret}] \geq \min \left\{ \frac{\E[T_{\ALG}]}{KL}, \frac{\E[T_{\ALG}]\Delta}{L} \right\}.
\end{align}
\end{lemma}

\begin{proof}
For any node $a \in \mT$, let $S_a \in [T_{\ALG}]$ be the first time that Nature visits $a$ and $E_a \in [T_{\ALG}]$ be the last time that Nature visits $a$.
If Nature never visits node $a$, then $S_a$ is defined as the time that Nature skips $a$, and $E_a = S_a -1$.
For any action $i\in [n]$, let $a(i) := \lfloor \frac{i-1}{2}\rfloor$ be the leaf node of $i$.
Define 
\begin{align*}
X_i = \sum_{t \in [S_{a(i)} - 1]}p_t(i); \qquad \quad Y_i = \sum_{t \in [S_{a(i)}: E_{a(i)}]}p_t(i); \qquad \quad 
Z_i = \sum_{t \in [E_{a(i)} + 1: T_{\ALG}]}p_t(i).
\end{align*}
That is, $X_i$ is the total probability mass that the algorithm places on $i$ before Nature visits the leaf node $a(i)$; $Y_i$ is the probability mass when Nature visits the leaf node $a(i)$; and $Z_i$ is the probability mass after visiting the leaf node $a(i)$.
By the definition, the total mass placed on action $i$ equals $X_i + Y_i + Z_i$ and one has $\sum_{i \in [n]}X_i + Y_i + Z_i = T_{\ALG}$.

We divide into three cases based on the value of $\sum_{i\in [n]}\E[X_i], \sum_{i\in [n]}\E[Y_i]$ and $\sum_{i\in [n]}\E[Z_i]$.

{\bf Case 1.} Suppose $\sum_{i \in [n]}\E[X_i] \geq \frac{1}{3}\E[T_{\ALG}]$. That is, the algorithm places large mass on actions before visiting their leaf nodes.

We first give an alternative way of computing the mass $\sum_{i \in [n]}X_i$. At level $\ell \in [0:L-1]$ and node $a \in \mT_\ell$, let $\N^{+}(a)$ contain all actions in the older siblings of $a$, i.e.,
\begin{align*}
\N^{+}(a) := \N_{a_{L}\ldots a_{\ell+2} (a_{\ell+1}+1)} \cup \cdots \cup \N_{a_{L}\ldots a_{\ell+2}K-1}.
\end{align*}
Note if $a$ is the oldest child node, i.e., $a_{\ell+1} = K-1$, then $\N^{+}(a) = \emptyset$.
Define
\begin{align}
M_a := \sum_{t \in [S_a: E_a]}\sum_{i \in \N^{+}(a)} p_t(i). \label{eq:m-a}
\end{align}
That is, $M_a$ is the total probability mass placed on $\N^{+}(a)$ (actions of older siblings of $a$) during the visit of node $a$. 
We make the following claim, whose proof can be found at Appendix \ref{sec:lower-app}.
\begin{lemma}
\label{lem:alternative}
We have $\sum_{i\in [n]}X_i = \sum_{\ell \in [0:L-1]}\sum_{a \in \mT_\ell}M_a$.
\end{lemma}
Let $\mV_\ell \subseteq \mT_\ell$ be the set of visited nodes at level $\ell$. 
Consider the following swap function $\phi$: For each level $\ell \in [0:L-1]$ and for each node $a \in \mT_{\ell}$ in level $\ell$, suppose (1) $a \in \mV_{\ell}$ has been visited and (2) its older sibling $a + 1 \notin \mV_{\ell}$ has been skipped, then the swap function maps actions in $\N^{+}_{a}$ to the last action in $\N_{a}$. 
It is easy to check that for every action $i \in [n]$, $\phi(i)$ is uniquely defined.

We can bound the swap regret as follow.

\begin{align}
\texttt{swap-regret} \geq &~ \sum_{i \in [n]} \sum_{t\in [T]}p_t(i)(r_t(\phi(i)) - r_t(i)) \notag \\
= &~ \sum_{\ell \in [0:L-1]}\sum_{a \in \mV_{\ell} \wedge (a+1) \notin \mV_{\ell}} \sum_{i \in \N^{+}(a)}\sum_{ t\in [T]}p_t(i)(r_t(\phi(i)) - r_t(i))\notag\\
= &~ \sum_{\ell \in [0:L-1]}\sum_{a \in \mV_{\ell} \wedge (a+1) \notin \mV_{\ell}} \sum_{i \in \N^{+}(a)}\sum_{t\in [S_a : E_a]}p_t(i)(r_t(\phi(i)) - r_t(i))\notag\\
\geq &~ \frac{1}{16(L+1)} \sum_{\ell \in [0:L-1]} \sum_{a \in \mV_{\ell} \wedge (a+1) \notin \mV_{\ell}} \sum_{i \in \N^{+}(a)}\sum_{t\in [S_a : E_a]}p_t(i)\notag\\
= &~ \frac{1}{16(L+1)} \sum_{\ell \in [0:L-1]}\sum_{a \in \mV_{\ell} \wedge (a+1) \notin \mV_{\ell}} M_a. \label{eq:case1-swap}
\end{align}

The second step holds since the swap function only changes actions in 
$
\bigcup_{\ell \in [0:L-1]}\bigcup_{a \in \mV_{\ell} \wedge (a+1)\notin \mV_{\ell}} \N^{+}(a).
$
The third step holds since the actions $i$ and $\phi(i)$ ($i\in \N^{+}(a)$) have different rewards only when Nature visits node $a$. The fourth step holds since 
\[
r_t(\phi(i)) - r_t(i) \geq \frac{L - \ell}{16(L+1)} - \frac{L - \ell -1}{16(L+1)} = \frac{1}{16(L+1)} \quad \forall t \in [S_a: E_a]
\]
according to the definition of $\phi$ and the reward sequence.
The last step holds by the definition of $M_a$ (see Eq.~\eqref{eq:m-a}).

For each level $\ell \in [0:L-1]$, we have
\begin{align}
\E\left[\sum_{a \in \mV_{\ell} \wedge (a+1) \notin \mV_{\ell}} M_a\right] = &~  \sum_{a \in \mT_\ell}\E\left[ M_a \cdot \mathsf{1}\{a \in \mV_{\ell} \wedge (a+1)\notin \mV_{\ell}\}\right]\notag \\
= &~ \sum_{a \in \mT_\ell}\E[ M_a | a \in \mV_{\ell} \wedge (a+1)\notin \mV_{\ell}] \cdot \Pr[a \in \mV_{\ell} \wedge (a+1)\notin \mV_{\ell}]\notag \\
= &~ \sum_{a \in \mT_\ell}\E[ M_a | a \in \mV_{\ell}] \cdot \Pr[a \in \mV_{\ell}] \cdot \frac{1}{2K}\notag  \\
= &~ \frac{1}{2K} \sum_{a \in \mT_\ell}\E[M_a] . \label{eq:case1-swap2}
\end{align}
The first step follows from the linearity of expectation and the second step follows from the law of expectation. 
The third step holds since for any node $a \in \mT_{\ell}$, condition on $a \in \mV_{\ell}$, the mass $M_a$ is independent of whether $(a+1)$ is skipped or not, and the node $(a + 1)$ is skipped with probability $q = \frac{1}{2L}$.
The fourth step holds since $\E[M_a | a \notin \mV_{\ell}] = 0$.

Taking an expectation over both sides of Eq.~\eqref{eq:case1-swap}, we have
\begin{align*}
\E[\texttt{swap-regret}] \geq &~ \frac{1}{16(L+1)} \E\left[\sum_{\ell\in [0:L-1]}\sum_{a \in \mV_{\ell} \wedge (a+1) \notin \mV_{\ell}} M_a\right]\\
\geq &~ \frac{1}{32K(L+1)} \sum_{\ell\in [0:L-1]}\sum_{a \in \mT_\ell}\E[M_a]  \\
= &~\frac{1}{32K(L+1)} \sum_{i \in [n]}\E[X_i]  = \Omega\left(\frac{\E[T_{\ALG}]}{KL}\right).
\end{align*}
The second step follows from Eq.~\eqref{eq:case1-swap2}, the third step follows from Lemma \ref{lem:alternative} and the last step follows from the assumption of the first case.

{\bf Case 2.} Suppose $\sum_{i \in [n]}\E[Y_i] \geq \frac{1}{3}\E[T_{\ALG}]$. That is, the algorithm spends a lot of time playing actions of the leaf node during its visit. 

Consider the following swap function $\phi$. For each leaf node $a \in [0:n/2-1]$, the swap function switches actions $2a+1, 2a+2$ to $2a + i^{*}(a)$, i.e., the action that draws reward from $\frac{L}{16(L+1)} + \frac{1}{16(L+1)} B_{1/2 + \Delta}$.

To bound the swap regret, we have
\begin{align}
\texttt{swap-regret} \geq &~ \sum_{t \in [T_{\ALG}]}\sum_{i\in [n]}p_t(i)r_t(\phi(i)) - p_t(i)r_t(i)\notag \\
= &~ \sum_{t\in [T_{\ALG}]}\sum_{a\in [0:n/2-1]}(p_t(2a+1) + p_t(2a+2))r_t(2a+i^{*}(a))\notag\\
&~ - \sum_{t\in [T_{\ALG}]}\sum_{a\in [0:n/2-1]} p_t(2a+1) r_t(2a+1) + p_t(2a+2)r_t(2a+2) \notag \\
= &~ \sum_{a \in [0:n/2-1]}\sum_{t \in [S_a: E_a]} (p_t(2a+1) + p_t(2a+2))r_t(2a+i^{*}(a))\notag \\
&~ - p_t(2a+1) r_t(2a+1) + p_t(2a+2)r_t(2a+2))\label{eq:case2-1}
\end{align}
The second step follows from the definition of our swap function, the third step holds since the actions $2a+1, 2a+2$ have the same reward except $t \in [S_a: E_a]$.

\paragraph{Technical component: Lower bound for two-coin game} 
In order to bound the RHS of Eq.~\eqref{eq:case2-1}, we consider an abstract problem which we call the {\em two-coin game}.
Let $\Delta \in (0,1/20), H = 1/400\Delta^2$ be input parameters.
In a two-coin game, there are two coins, one draws from the Bernoulli distribution $B_{1/2}$ and the other draws from $B_{1/2 + \Delta}$. The biased coin $i^{*}$ is chosen uniformly at random and it is not known to the player.

The two-coin game is repeatedly played for $H$ days.
At each day $h \in [H]$, the player commits a distribution $p_h \in \Delta_3$ over coin $1$, coin $2$ and {\em a dummy action}.
The dummy action is interpreted as an outside option, aka not playing among the two coins. 
It then samples from the two coins and observes the reward $r_h\in \{0, 1\}^2$.
The following Lemma bounds the regret of switching between two coins and its proof is deferred to Appendix \ref{sec:lower-app}.

\begin{lemma}[Lower bound for two-coin game]
\label{lem:tech}
In a two-coin game, the expected swap regret of switching between two coins satisfy 
\begin{align*}
&~ \E\left[\sum_{h \in [H]} (p_h(1) + p_h(2))r_{h}(i^*) - \sum_{h\in [H]}(p_{h}(1)r_{h}(1) + p_{h}(2)r_{h}(2))\right]\\
\geq &~ \Delta \cdot \left(\frac{1}{2}\sum_{h\in [H]}\E[p_h(1) + p_h(2)] - \frac{3}{20}H\right).
\end{align*}
Here the expectation is taken over the randomness of the reward and the algorithm.
\end{lemma}

Now we are about to use Lemma \ref{lem:tech}. For each leaf node $a \in [0:n/2-1]$, if Nature visits leaf $a$, then during the time $[S_a: E_a]$, one can view Nature and the algorithm play a two-coin game, where the two coins are $2a+1, 2a+2$ and the dummy action includes the rest of actions in $[n]\backslash \{2a+1, 2a + 2\}$. They are the same up to a common offset of $\frac{L}{16(L+1)}$ and a scaling factor of $\frac{1}{16(L+1)}$. Hence, for a fixed leaf node $a$, we have
\begin{align}
&~ \E\left[\sum_{t \in [S_a: E_a]} (p_t(2a+1) + p_t(2a+2))r_t(2a+i^{*}(a)) - p_t(2a+1) r_t(2a+1) - p_t(2a+2)r_t(2a+2)\right]\notag \\
\geq &~\frac{1}{16(L+1)} \cdot \Delta \left(\frac{1}{2} \E[Y_{2a+1} + Y_{2a+2}] - \frac{3}{20} \E[\mathsf{1}\{a \in \mV_{0}\}] \cdot H\right) \label{eq:case2-2}
\end{align}
where we apply Lemma \ref{lem:tech}.

Combining Eq.~\eqref{eq:case2-1}\eqref{eq:case2-2}, the expected swap regret is at least
\begin{align*}
\E[\texttt{swap-regret}] \geq &~ \frac{\Delta}{16(L+1)}\left(\frac{1}{2}\sum_{a\in [0:n/2-1]}\E[Y_{2a+1}+ Y_{2a+2}] - \frac{3}{20}\sum_{a\in [0:n/2-1]}\E[\mathsf{1}\{a \in \mV_{0}\}] \cdot H\right)\\
= &~ \Omega\left(\frac{\E[T_{\ALG}]\Delta}{L}\right).
\end{align*}
Here we use the fact that $\sum_{a\in [0:n/2-1]}\E[\mathsf{1}\{a \in \mV_{0}\}]\cdot H = \E[T_{\ALG}]$ and our assumption $\sum_{i\in [n]}\E[Y_i]\geq \frac{1}{3}\E[T_{\ALG}]$.

{\bf Case 3.} Suppose $\sum_{i \in [n]}\E[Z_i] \geq \frac{1}{3}\E[T_{\ALG}]$. That is, the algorithm spends a lot of time playing actions that have already been visited. In this case, it suffices to switch to the fixed action $n$. 
\begin{align*}
\texttt{swap-regret} \geq &~ \sum_{t\in [T_{\ALG}]}r_t(n) - \sum_{i\in [n]} \sum_{t\in [T_{\ALG}]}p_t(i)r_t(i) \\
= &~ \sum_{t\in [T_{\ALG}]}r_t(n) - \sum_{i\in [n]}\left(\sum_{t \in [E_{a(i)}]} p_t(i)r_t(i) + \sum_{t \in [E_{a(i)} + 1: T_{\ALG}]} p_t(i)r_t(i)\right) \\
\geq &~ 0 - \sum_{i\in [n]}\left((X_i + Y_i)\cdot \frac{1}{16} + Z_i \cdot (-1)\right)\\
\geq &~ \frac{17}{16}\sum_{i\in [n]}Z_i - \frac{1}{16}T_{\ALG}.
\end{align*}
The third step follows from the maximum reward is $\frac{1}{16}$ and the reward of action $i$ is $-1$ after $E_{a(i)}$

Taking an expectation, the expected swap regret is at least $\frac{1}{4}\E[T_{\ALG}]$ in Case 3.

Combing the above three cases, we have finish the proof of Lemma \ref{lem:lower-goal}.
\end{proof}

The sequence length $T_{\ALG}$ is a random variable, and its expectation satisfies
\begin{lemma}
\label{lem:basic-expecation}
Let $C_{K} = \sum_{k=0}^{K-1}(1-\frac{1}{2K})^{k}$.
We have
\[
\E[T_{\ALG}] = H \cdot (C_K)^{L}  \geq 2^{-L} \cdot \frac{K^{L}}{400\Delta^{2}}.
\]
\end{lemma}
The proof can be found at Appendix \ref{sec:lower-app} and we can now prove Theorem \ref{thm:lower-oblivious}.
\begin{proof}[Proof of Theorem \ref{thm:lower-oblivious}]
Recall the parameters $K, L$ are chosen such that the number of actions $n = 2K^{L}$. 
For any fixed constant $\delta > 0$, we take $K = 2^{1/\delta}$ and $L = \delta(\log_2(n) - 1)$. We prove the expected swap regret over $T$ days is at least 
\[
\E[\texttt{swap-regret}] \geq \Omega\left(\min\left\{\frac{T}{\log(T)},   \sqrt{n^{1-8\delta} T} \right\}\right)
\]
Taking $\delta \rightarrow 0$ would be sufficient for our proof.

First, if $T \geq n$, then take $\Delta = \sqrt{n/400T} \leq \frac{1}{20}$ and consider the hard sequence of Algorithm \ref{algo:hard-sequence-oblivious}. Note the maximum sequence length $T_{\ALG} \leq \frac{K^{L}}{400\Delta^2} = T/2$, and for the last $T - T_{\ALG}$ days, the reward vector is taken to be all $0$. By Lemma \ref{lem:lower-goal}, the total regret is at least 
\begin{align}
\E[\texttt{swap-regret}] \geq &~ \min \left\{ \frac{\E[T_{\ALG}]}{KL}, \frac{\E[T_{\ALG}]\Delta}{L} \right\} \geq 2^{-L} \frac{K^{L}}{400\Delta^2} \min\left\{\frac{1}{KL}, \frac{\Delta}{L} \right\}\notag\\
\geq &~ \Omega(n^{1/2-3\delta/2}T^{1/2}).  \label{eq:lower-large}
\end{align}
The second step follows from Lemma \ref{lem:basic-expecation} and the last step follows from the choice of parameters.

Second, if $T \in [n^{1-2\delta}, n]$, then we claim the swap regret has to be least $n^{1-4\delta} \geq \sqrt{n^{1-8\delta}T}$. Otherwise, consider the algorithm that restarts every $T$ days, its swap regret over $n$ rounds is at most $n^{1-4\delta} \cdot \lceil n/T \rceil \leq n^{1-2\delta}$, this contradicts with Eq.~\eqref{eq:lower-large}. 

Third, if $T = n^{1-2\delta}$, the we prove the swap regret is at least $\Omega(n^{1-2\delta}/\log(n)) = \Omega(T/\log(T))$. We prove by contradiction. Suppose there is an algorithm that has swap regret at most $o(n^{1-2\delta}/\log(n))$ over $n^{1-2\delta}$ days. Then for any $T'$, there is an algorithm that has swap regret at most 
$
\lceil T'/n^{1-2\delta} \rceil \cdot o(n^{1-2\delta}/\log(n)) \leq o(T'/\log(n)) + n^{1-2\delta}
$
over $T'$ days (without knowing $T'$ in advance), as one can always restart the algorithm every $T = n^{1-2\delta}$ days.
Applying this algorithm to the hard sequence with $\Delta = 1/20$, its swap regret is at most $o(\E[T_{\ALG}]/\log(n)+n^{1-2\delta}) = o(\E[T_{\ALG}]/\log(n))$.  
However, by Lemma \ref{lem:lower-goal}, any algorithm must suffer swap regret at least
$\min \left\{ \frac{\E[T_{\ALG}]}{KL}, \frac{\E[T_{\ALG}]\Delta}{L} \right\} = \Omega(\E[T_{\ALG}]/\log(n))$.
This reaches a contradiction.

Finally, if $T \leq n^{1-2\delta}$. One can merge $n/T^{1/(1-2\delta)}$ actions into one action by assigning the same reward to them. Then the swap regret is at least $T/\log(T)$ by the third case. We complete the proof here.
\end{proof}

	\bibliographystyle{alpha}
	\bibliography{ref}

\newcommand{\etalchar}[1]{$^{#1}$}
\begin{thebibliography}{CBFH{\etalchar{+}}97}

\bibitem[AAD{\etalchar{+}}23]{assos2023online}
Angelos Assos, Idan Attias, Yuval Dagan, Constantinos Daskalakis, and
  Maxwell~K. Fishelson.
\newblock Online learning and solving infinite games with an {ERM} oracle.
\newblock In Gergely Neu and Lorenzo Rosasco, editors, {\em The Thirty Sixth
  Annual Conference on Learning Theory, {COLT} 2023, 12-15 July 2023,
  Bangalore, India}, volume 195 of {\em Proceedings of Machine Learning
  Research}, pages 274--324. {PMLR}, 2023.

\bibitem[ADF{\etalchar{+}}22]{anagnostides2022near}
Ioannis Anagnostides, Constantinos Daskalakis, Gabriele Farina, Maxwell
  Fishelson, Noah Golowich, and Tuomas Sandholm.
\newblock Near-optimal no-regret learning for correlated equilibria in
  multi-player general-sum games.
\newblock In {\em Proceedings of the 54th Annual ACM SIGACT Symposium on Theory
  of Computing}, pages 736--749, 2022.

\bibitem[AFK{\etalchar{+}}22a]{anagnostides2022faster}
Ioannis Anagnostides, Gabriele Farina, Christian Kroer, Andrea Celli, Tuomas
  Sandholm, et~al.
\newblock Faster no-regret learning dynamics for extensive-form correlated and
  coarse correlated equilibria.
\newblock In {\em EC'22: Proceedings of the 23rd ACM Conference on Economics
  and Computation}, 2022.

\bibitem[AFK{\etalchar{+}}22b]{anagnostides2022uncoupled}
Ioannis Anagnostides, Gabriele Farina, Christian Kroer, Chung-Wei Lee, Haipeng
  Luo, and Tuomas Sandholm.
\newblock Uncoupled learning dynamics with {$O(\log T)$} swap regret in
  multiplayer games.
\newblock {\em Advances in Neural Information Processing Systems},
  35:3292--3304, 2022.

\bibitem[AFS23]{anagnostides2023near}
Ioannis Anagnostides, Gabriele Farina, and Tuomas Sandholm.
\newblock Near-optimal {$\Phi$}-regret learning in extensive-form games.
\newblock In {\em International Conference on Machine Learning}, pages
  814--839. PMLR, 2023.

\bibitem[AHK12]{arora2012multiplicative}
Sanjeev Arora, Elad Hazan, and Satyen Kale.
\newblock The multiplicative weights update method: a meta-algorithm and
  applications.
\newblock {\em Theory of computing}, 8(1):121--164, 2012.

\bibitem[Aum74]{aumann1974subjectivity}
Robert~J Aumann.
\newblock Subjectivity and correlation in randomized strategies.
\newblock {\em Journal of mathematical Economics}, 1(1):67--96, 1974.

\bibitem[Bab16]{babichenko2016query}
Yakov Babichenko.
\newblock Query complexity of approximate nash equilibria.
\newblock {\em Journal of the ACM (JACM)}, 63(4):1--24, 2016.

\bibitem[Bab20]{babichenko2020informational}
Yakov Babichenko.
\newblock Informational bounds on equilibria (a survey).
\newblock {\em ACM SIGecom Exchanges}, 17(2):25--45, 2020.

\bibitem[BB15]{babichenko2015query}
Yakov Babichenko and Siddharth Barman.
\newblock Query complexity of correlated equilibrium.
\newblock {\em ACM Transactions on Economics and Computation (TEAC)},
  3(4):1--9, 2015.

\bibitem[BDN19]{babichenko2019communication}
Yakov Babichenko, Shahar Dobzinski, and Noam Nisan.
\newblock The communication complexity of local search.
\newblock In {\em Proceedings of the 51st Annual ACM SIGACT Symposium on Theory
  of Computing}, pages 650--661, 2019.

\bibitem[BJM{\etalchar{+}}22]{bai2022efficient}
Yu~Bai, Chi Jin, Song Mei, Ziang Song, and Tiancheng Yu.
\newblock Efficient phi-regret minimization in extensive-form games via online
  mirror descent.
\newblock {\em Advances in Neural Information Processing Systems},
  35:22313--22325, 2022.

\bibitem[BLGS19]{brown2019deep}
Noam Brown, Adam Lerer, Sam Gross, and Tuomas Sandholm.
\newblock Deep counterfactual regret minimization.
\newblock In {\em International conference on machine learning}, pages
  793--802. PMLR, 2019.

\bibitem[BM07]{blum2007external}
Avrim Blum and Yishay Mansour.
\newblock From external to internal regret.
\newblock {\em Journal of Machine Learning Research}, 8(6), 2007.

\bibitem[BMSW18]{braverman2018selling}
Mark Braverman, Jieming Mao, Jon Schneider, and Matt Weinberg.
\newblock Selling to a no-regret buyer.
\newblock In {\em Proceedings of the 2018 ACM Conference on Economics and
  Computation}, pages 523--538, 2018.

\bibitem[BR17]{babichenko2017communication}
Yakov Babichenko and Aviad Rubinstein.
\newblock Communication complexity of approximate nash equilibria.
\newblock In {\em Proceedings of the 49th Annual ACM SIGACT Symposium on Theory
  of Computing}, pages 878--889, 2017.

\bibitem[BR20]{babichenko2020communication}
Yakov Babichenko and Aviad Rubinstein.
\newblock Communication complexity of nash equilibrium in potential games.
\newblock In {\em 2020 IEEE 61st Annual Symposium on Foundations of Computer
  Science (FOCS)}, pages 1439--1445. IEEE, 2020.

\bibitem[Bro51]{brown1951iterative}
George~W. Brown.
\newblock Iterative solutions of games by fictitious play.
\newblock {\em Activity Analysis of Production and Allocation}, 1951.

\bibitem[BS18]{brown2018superhuman}
Noam Brown and Tuomas Sandholm.
\newblock Superhuman ai for heads-up no-limit poker: Libratus beats top
  professionals.
\newblock {\em Science}, 359(6374):418--424, 2018.

\bibitem[BS19]{brown2019superhuman}
Noam Brown and Tuomas Sandholm.
\newblock Superhuman ai for multiplayer poker.
\newblock {\em Science}, 365(6456):885--890, 2019.

\bibitem[BSV23]{brown2023learning}
William Brown, Jon Schneider, and Kiran Vodrahalli.
\newblock Is learning in games good for the learners?
\newblock {\em Advances in Neural Information Processing Systems}, 2023.

\bibitem[CBFH{\etalchar{+}}97]{cesa1997use}
Nicolo Cesa-Bianchi, Yoav Freund, David Haussler, David~P Helmbold, Robert~E
  Schapire, and Manfred~K Warmuth.
\newblock How to use expert advice.
\newblock {\em Journal of the ACM (JACM)}, 44(3):427--485, 1997.

\bibitem[CBL03]{cesa2003potential}
Nicolo Cesa-Bianchi and G{\'a}bor Lugosi.
\newblock Potential-based algorithms in on-line prediction and game theory.
\newblock {\em Machine Learning}, 51:239--261, 2003.

\bibitem[CBL06]{cesa2006prediction}
Nicolo Cesa-Bianchi and G{\'a}bor Lugosi.
\newblock {\em Prediction, learning, and games}.
\newblock Cambridge university press, 2006.

\bibitem[CCG19]{celli19computing}
Andrea Celli, Stefano Coniglio, and Nicola Gatti.
\newblock Computing optimal ex ante correlated equilibria in two-player
  sequential games.
\newblock In Edith Elkind, Manuela Veloso, Noa Agmon, and Matthew~E. Taylor,
  editors, {\em Proceedings of the 18th International Conference on Autonomous
  Agents and MultiAgent Systems, {AAMAS} '19, Montreal, QC, Canada, May 13-17,
  2019}, pages 909--917. International Foundation for Autonomous Agents and
  Multiagent Systems, 2019.

\bibitem[CCT15]{chen2015well}
Xi~Chen, Yu~Cheng, and Bo~Tang.
\newblock Well-supported versus approximate nash equilibria: Query complexity
  of large games.
\newblock In {\em Innovations in Theoretical Computer Science (ITCS)}, 2015.

\bibitem[CDT09]{chen2009settling}
Xi~Chen, Xiaotie Deng, and Shang-Hua Teng.
\newblock Settling the complexity of computing two-player nash equilibria.
\newblock {\em Journal of the ACM (JACM)}, 56(3):1--57, 2009.

\bibitem[CHJ20]{camara2020mechanisms}
Modibo~K Camara, Jason~D Hartline, and Aleck Johnsen.
\newblock Mechanisms for a no-regret agent: Beyond the common prior.
\newblock In {\em 2020 ieee 61st annual symposium on foundations of computer
  science (focs)}, pages 259--270. IEEE, 2020.

\bibitem[CP20]{chen2020hedging}
Xi~Chen and Binghui Peng.
\newblock Hedging in games: Faster convergence of external and swap regrets.
\newblock {\em Advances in Neural Information Processing Systems},
  33:18990--18999, 2020.

\bibitem[CSK23]{chhablani2023multiplicative}
Chirag Chhablani, Michael Sullins, and Ian~A Kash.
\newblock Multiplicative weight updates for extensive form games.
\newblock In {\em Proceedings of the 2023 International Conference on
  Autonomous Agents and Multiagent Systems}, pages 1071--1078, 2023.

\bibitem[CWWZ23]{cai2023selling}
Linda Cai, S~Matthew Weinberg, Evan Wildenhain, and Shirley Zhang.
\newblock Selling to multiple no-regret buyers.
\newblock {\em arXiv preprint arXiv:2307.04175}, 2023.

\bibitem[Daw82]{dawid1982well}
A~Philip Dawid.
\newblock The well-calibrated bayesian.
\newblock {\em Journal of the American Statistical Association},
  77(379):605--610, 1982.

\bibitem[DDFG23]{dagan2023external}
Yuval Dagan, Constantinos Daskalakis, Maxwell Fishelson, and Noah Golowich.
\newblock From external to swap regret 2.0: An efficient reduction and
  oblivious adversary for large action spaces.
\newblock {\em arXiv preprint arXiv:2310.19786}, 2023.

\bibitem[DDK11]{daskalakis2011near}
Constantinos Daskalakis, Alan Deckelbaum, and Anthony Kim.
\newblock Near-optimal no-regret algorithms for zero-sum games.
\newblock In {\em Proceedings of the twenty-second annual ACM-SIAM symposium on
  Discrete Algorithms}, pages 235--254. SIAM, 2011.

\bibitem[DFG21]{daskalakis2021near}
Constantinos Daskalakis, Maxwell Fishelson, and Noah Golowich.
\newblock Near-optimal no-regret learning in general games.
\newblock {\em Advances in Neural Information Processing Systems},
  34:27604--27616, 2021.

\bibitem[DG22]{daskalakis2022fast}
Constantinos Daskalakis and Noah Golowich.
\newblock Fast rates for nonparametric online learning: from realizability to
  learning in games.
\newblock In {\em Proceedings of the 54th Annual ACM SIGACT Symposium on Theory
  of Computing}, pages 846--859, 2022.

\bibitem[DGP09]{daskalakis2009complexity}
Constantinos Daskalakis, Paul~W Goldberg, and Christos~H Papadimitriou.
\newblock The complexity of computing a nash equilibrium.
\newblock {\em Communications of the ACM}, 52(2):89--97, 2009.

\bibitem[DSS19a]{deng2019prior}
Yuan Deng, Jon Schneider, and Balasubramanian Sivan.
\newblock Prior-free dynamic auctions with low regret buyers.
\newblock {\em Advances in Neural Information Processing Systems}, 32, 2019.

\bibitem[DSS19b]{deng2019strategizing}
Yuan Deng, Jon Schneider, and Balasubramanian Sivan.
\newblock Strategizing against no-regret learners.
\newblock {\em Advances in neural information processing systems}, 32, 2019.

\bibitem[FAL{\etalchar{+}}22]{farina2022near}
Gabriele Farina, Ioannis Anagnostides, Haipeng Luo, Chung-Wei Lee, Christian
  Kroer, and Tuomas Sandholm.
\newblock Near-optimal no-regret learning dynamics for general convex games.
\newblock {\em Advances in Neural Information Processing Systems},
  35:39076--39089, 2022.

\bibitem[FCMG22]{farina2022simple}
Gabriele Farina, Andrea Celli, Alberto Marchesi, and Nicola Gatti.
\newblock Simple uncoupled no-regret learning dynamics for extensive-form
  correlated equilibrium.
\newblock {\em Journal of the ACM}, 69(6):1--41, 2022.

\bibitem[FGL{\etalchar{+}}21]{feng2021convergence}
Zhe Feng, Guru Guruganesh, Christopher Liaw, Aranyak Mehta, and Abhishek Sethi.
\newblock Convergence analysis of no-regret bidding algorithms in repeated
  auctions.
\newblock In {\em Proceedings of the AAAI Conference on Artificial
  Intelligence}, volume~35, pages 5399--5406, 2021.

\bibitem[FKS19a]{farina2019optimistic}
Gabriele Farina, Christian Kroer, and Tuomas Sandholm.
\newblock Optimistic regret minimization for extensive-form games via dilated
  distance-generating functions.
\newblock {\em Advances in neural information processing systems}, 32, 2019.

\bibitem[FKS19b]{farina2019regret}
Gabriele Farina, Christian Kroer, and Tuomas Sandholm.
\newblock Regret circuits: Composability of regret minimizers.
\newblock In {\em International conference on machine learning}, pages
  1863--1872. PMLR, 2019.

\bibitem[FL99]{fudenberg1999conditional}
Drew Fudenberg and David~K. Levine.
\newblock Conditional universal consistency.
\newblock {\em Games and Economic Behavior}, 29(1):104--130, 1999.

\bibitem[FLFS19]{farina2019efficient}
Gabriele Farina, Chun~Kai Ling, Fei Fang, and Tuomas Sandholm.
\newblock Efficient regret minimization algorithm for extensive-form correlated
  equilibrium.
\newblock {\em Advances in Neural Information Processing Systems}, 32, 2019.

\bibitem[FLL{\etalchar{+}}16]{foster2016learning}
Dylan~J Foster, Zhiyuan Li, Thodoris Lykouris, Karthik Sridharan, and Eva
  Tardos.
\newblock Learning in games: Robustness of fast convergence.
\newblock {\em Advances in Neural Information Processing Systems}, 29, 2016.

\bibitem[FLLK22]{farina2022kernelized}
Gabriele Farina, Chung-Wei Lee, Haipeng Luo, and Christian Kroer.
\newblock Kernelized multiplicative weights for 0/1-polyhedral games: Bridging
  the gap between learning in extensive-form and normal-form games.
\newblock In {\em International Conference on Machine Learning}, pages
  6337--6357. PMLR, 2022.

\bibitem[For93]{forges1993five}
Fran{\c{c}}oise Forges.
\newblock Five legitimate definitions of correlated equilibrium in games with
  incomplete information.
\newblock {\em Theory and decision}, 35:277--310, 1993.

\bibitem[FP23]{farina2023polynomial}
Gabriele Farina and Charilaos Pipis.
\newblock Polynomial-time linear-swap regret minimization in
  imperfect-information sequential games.
\newblock {\em Advances in Neural Information Processing Systems}, 2023.

\bibitem[FS97]{freund1997decision}
Yoav Freund and Robert~E Schapire.
\newblock A decision-theoretic generalization of on-line learning and an
  application to boosting.
\newblock {\em Journal of computer and system sciences}, 55(1):119--139, 1997.

\bibitem[FS99]{freund1999adaptive}
Yoav Freund and Robert~E Schapire.
\newblock Adaptive game playing using multiplicative weights.
\newblock {\em Games and Economic Behavior}, 29(1-2):79--103, 1999.

\bibitem[Fuj23]{fujii2023bayes}
Kaito Fujii.
\newblock Bayes correlated equilibria and no-regret dynamics.
\newblock {\em arXiv preprint arXiv:2304.05005}, 2023.

\bibitem[FV93]{foster1993randomization}
Dean~P Foster and Rakesh~V Vohra.
\newblock A randomization rule for selecting forecasts.
\newblock {\em Operations Research}, 41(4):704--709, 1993.

\bibitem[FV97]{foster1997calibrated}
Dean~P Foster and Rakesh~V Vohra.
\newblock Calibrated learning and correlated equilibrium.
\newblock {\em Games and Economic Behavior}, 21(1-2):40, 1997.

\bibitem[FV98]{foster1998asymptotic}
Dean~P Foster and Rakesh~V Vohra.
\newblock Asymptotic calibration.
\newblock {\em Biometrika}, 85(2):379--390, 1998.

\bibitem[FV99]{foster1999regret}
Dean~P Foster and Rakesh Vohra.
\newblock Regret in the on-line decision problem.
\newblock {\em Games and Economic Behavior}, 29(1-2):7--35, 1999.

\bibitem[GC18]{ganor2018communication}
Anat Ganor and Karthik CS.
\newblock Communication complexity of correlated equilibrium with small
  support.
\newblock In {\em Approximation, Randomization, and Combinatorial Optimization.
  Algorithms and Techniques (APPROX/RANDOM 2018)}. Schloss
  Dagstuhl-Leibniz-Zentrum fuer Informatik, 2018.

\bibitem[GP21]{ganor2021communication}
Anat Ganor and D{\"o}m{\"o}t{\"o}r P{\'a}lv{\"o}lgyi.
\newblock On communication complexity of fixed point computation.
\newblock {\em ACM transactions on economics and computation}, 9(4):1--27,
  2021.

\bibitem[GR16]{goldberg2016bounds}
Paul~W Goldberg and Aaron Roth.
\newblock Bounds for the query complexity of approximate equilibria.
\newblock {\em ACM Transactions on Economics and Computation (TEAC)},
  4(4):1--25, 2016.

\bibitem[GR18]{goos2018near}
Mika G{\"o}{\"o}s and Aviad Rubinstein.
\newblock Near-optimal communication lower bounds for approximate nash
  equilibria.
\newblock In {\em 2018 IEEE 59th Annual Symposium on Foundations of Computer
  Science (FOCS)}, pages 397--403. IEEE, 2018.

\bibitem[HM10]{hart2010long}
Sergiu Hart and Yishay Mansour.
\newblock How long to equilibrium? the communication complexity of uncoupled
  equilibrium procedures.
\newblock {\em Games and Economic Behavior}, 69(1):107--126, 2010.

\bibitem[HMC00]{hart2000simple}
Sergiu Hart and Andreu Mas-Colell.
\newblock A simple adaptive procedure leading to correlated equilibrium.
\newblock {\em Econometrica}, 68(5):1127--1150, 2000.

\bibitem[HMC01]{hart2001reinforcement}
Sergiu Hart and Andreu Mas-Colell.
\newblock A reinforcement procedure leading to correlated equilibrium.
\newblock In {\em Economics Essays: A Festschrift for Werner Hildenbrand},
  pages 181--200. Springer, 2001.

\bibitem[HMC03]{hart2003uncoupled}
Sergiu Hart and Andreu Mas-Colell.
\newblock Uncoupled dynamics do not lead to nash equilibrium.
\newblock {\em American Economic Review}, 93(5):1830--1836, 2003.

\bibitem[HMC13]{hart2013simple}
Sergiu Hart and Andreu Mas-Colell.
\newblock {\em Simple adaptive strategies: from regret-matching to uncoupled
  dynamics}, volume~4.
\newblock World Scientific, 2013.

\bibitem[HN18]{hart2018query}
Sergiu Hart and Noam Nisan.
\newblock The query complexity of correlated equilibria.
\newblock {\em Games and Economic Behavior}, 108:401--410, 2018.

\bibitem[HPY23]{haghtalab2023calibrated}
Nika Haghtalab, Chara Podimata, and Kunhe Yang.
\newblock Calibrated stackelberg games: Learning optimal commitments against
  calibrated agents.
\newblock {\em Advances in Neural Information Processing Systems}, 2023.

\bibitem[HS89]{hart1989existence}
Sergiu Hart and David Schmeidler.
\newblock Existence of correlated equilibria.
\newblock {\em Mathematics of Operations Research}, 14(1):18--25, 1989.

\bibitem[HST15]{hartline2015no}
Jason Hartline, Vasilis Syrgkanis, and Eva Tardos.
\newblock No-regret learning in bayesian games.
\newblock {\em Advances in Neural Information Processing Systems}, 28, 2015.

\bibitem[HvS08]{huang2008computing}
Wan Huang and Bernhard von Stengel.
\newblock Computing an extensive-form correlated equilibrium in polynomial
  time.
\newblock In {\em International Workshop on Internet and Network Economics},
  pages 506--513. Springer, 2008.

\bibitem[Ito20]{ito2020tight}
Shinji Ito.
\newblock A tight lower bound and efficient reduction for swap regret.
\newblock {\em Advances in Neural Information Processing Systems},
  33:18550--18559, 2020.

\bibitem[JLB15]{jiang2015polynomial}
Albert~Xin Jiang and Kevin Leyton-Brown.
\newblock Polynomial-time computation of exact correlated equilibrium in
  compact games.
\newblock {\em Games and Economic Behavior}, 91:347--359, 2015.

\bibitem[KV05]{kalai2005efficient}
Adam Kalai and Santosh Vempala.
\newblock Efficient algorithms for online decision problems.
\newblock {\em Journal of Computer and System Sciences}, 71(3):291--307, 2005.

\bibitem[LW94]{littlestone1994weighted}
Nick Littlestone and Manfred~K Warmuth.
\newblock The weighted majority algorithm.
\newblock {\em Information and computation}, 108(2):212--261, 1994.

\bibitem[LWZB09]{lanctot2009monte}
Marc Lanctot, Kevin Waugh, Martin Zinkevich, and Michael Bowling.
\newblock Monte carlo sampling for regret minimization in extensive games.
\newblock {\em Advances in neural information processing systems}, 22, 2009.

\bibitem[MMSS22]{mansour2022strategizing}
Yishay Mansour, Mehryar Mohri, Jon Schneider, and Balasubramanian Sivan.
\newblock Strategizing against learners in bayesian games.
\newblock In {\em Conference on Learning Theory}, pages 5221--5252. PMLR, 2022.

\bibitem[MPPS23]{milionis2023impossibility}
Jason Milionis, Christos Papadimitriou, Georgios Piliouras, and Kelly
  Spendlove.
\newblock An impossibility theorem in game dynamics.
\newblock {\em Proceedings of the National Academy of Sciences},
  120(41):e2305349120, 2023.

\bibitem[Nas50]{nash1950equilibrium}
John Nash.
\newblock Equilibrium points in n-person games.
\newblock {\em Proceedings of the national academy of sciences}, 36(1):48--49,
  1950.

\bibitem[Nas51]{nash1951non}
John Nash.
\newblock Non-cooperative games.
\newblock {\em Annals of mathematics}, pages 286--295, 1951.

\bibitem[NRTV07]{nisan2007algorithmic}
Noam Nisan, Tim Roughgarden, Eva Tardos, and Vijay~V Vazirani.
\newblock Algorithmic game theory, 2007.
\newblock {\em Book available for free online}, 2007.

\bibitem[PR08]{papadimitriou2008computing}
Christos~H Papadimitriou and Tim Roughgarden.
\newblock Computing correlated equilibria in multi-player games.
\newblock {\em Journal of the ACM (JACM)}, 55(3):1--29, 2008.

\bibitem[Rob51]{robinson1951iterative}
Julia Robinson.
\newblock An iterative method of solving a game.
\newblock {\em Annals of Mathematics}, 54:296--301, 1951.

\bibitem[RS13a]{rakhlin2013online}
Alexander Rakhlin and Karthik Sridharan.
\newblock Online learning with predictable sequences.
\newblock In {\em Conference on Learning Theory}, pages 993--1019. PMLR, 2013.

\bibitem[RS13b]{rakhlin2013optimization}
Sasha Rakhlin and Karthik Sridharan.
\newblock Optimization, learning, and games with predictable sequences.
\newblock {\em Advances in Neural Information Processing Systems}, 26, 2013.

\bibitem[Rub15]{rubinstein2015inapproximability}
Aviad Rubinstein.
\newblock Inapproximability of nash equilibrium.
\newblock In {\em Proceedings of the forty-seventh annual ACM symposium on
  Theory of computing}, pages 409--418, 2015.

\bibitem[Rub16]{rubinstein2016settling}
Aviad Rubinstein.
\newblock Settling the complexity of computing approximate two-player nash
  equilibria.
\newblock In {\em 2016 IEEE 57th Annual Symposium on Foundations of Computer
  Science (FOCS)}, pages 258--265. IEEE, 2016.

\bibitem[RW16]{roughgarden2016communication}
Tim Roughgarden and Omri Weinstein.
\newblock On the communication complexity of approximate fixed points.
\newblock In {\em 2016 IEEE 57th Annual Symposium on Foundations of Computer
  Science (FOCS)}, pages 229--238. IEEE, 2016.

\bibitem[SALS15]{syrgkanis2015fast}
Vasilis Syrgkanis, Alekh Agarwal, Haipeng Luo, and Robert~E Schapire.
\newblock Fast convergence of regularized learning in games.
\newblock {\em Advances in Neural Information Processing Systems}, 28, 2015.

\bibitem[SL05]{stoltz2005internal}
Gilles Stoltz and G{\'a}bor Lugosi.
\newblock Internal regret in on-line portfolio selection.
\newblock {\em Machine Learning}, 59:125--159, 2005.

\bibitem[SL07]{stoltz2007learning}
Gilles Stoltz and G{\'a}bor Lugosi.
\newblock Learning correlated equilibria in games with compact sets of
  strategies.
\newblock {\em Games and Economic Behavior}, 59(1):187--208, 2007.

\bibitem[VSF08]{von2008extensive}
Bernhard Von~Stengel and Fran{\c{c}}oise Forges.
\newblock Extensive-form correlated equilibrium: Definition and computational
  complexity.
\newblock {\em Mathematics of Operations Research}, 33(4):1002--1022, 2008.

\bibitem[ZFCS22]{zhang2022optimal}
Brian~Hu Zhang, Gabriele Farina, Andrea Celli, and Tuomas Sandholm.
\newblock Optimal correlated equilibria in general-sum extensive-form games:
  Fixed-parameter algorithms, hardness, and two-sided column-generation.
\newblock In {\em Proceedings of the 23rd ACM conference on economics and
  computation}, pages 1119--1120, 2022.

\bibitem[ZJBP07]{zinkevich2007regret}
Martin Zinkevich, Michael Johanson, Michael Bowling, and Carmelo Piccione.
\newblock Regret minimization in games with incomplete information.
\newblock {\em Advances in neural information processing systems}, 20, 2007.

\bibitem[ZS22]{zhang2022polynomial}
Brian Zhang and Tuomas Sandholm.
\newblock Polynomial-time optimal equilibria with a mediator in extensive-form
  games.
\newblock {\em Advances in Neural Information Processing Systems},
  35:24851--24863, 2022.

\end{thebibliography}

	\newpage
	\appendix
	\section{A historical remark on internal vs swap regret} \label{app:internal-regret}
Our notion of swap regret can equivalently be written as:
\begin{align*}
\texttt{swap-regret} \text{ (equivalent)}:= \sum_{i\in [n]}\max_{\phi(i) \in [n]}\sum_{t\in [T]} \E_{i \sim p_t} [r_t(\phi(i)) - r_t(i)].
\end{align*}

The origin of the closely related ``internal regret'' is usually cited to~\cite{foster1998asymptotic}. Their notion of regret is almost identical to our swap regret, except that they take the expectation over the algorithm's randomness outside the summation. For oblivious adversaries this is equivalent, but~\cite{blum2007external,ito2020tight} prove lower bounds on this notion of swap/internal regret using an adaptive adversary that makes the algorithm regret its realized actions. Our work shows that these lower bounds do not extend to the distributional setting.
\begin{align*}
\texttt{internal-regret} \text{~\cite{foster1998asymptotic}}:= \E_{\text{algorithm's randomness}}[\sum_{i\in [n]}\max_{\phi(i) \in [n]}\sum_{t\in [T]} \id_{\{p^t = i\}}\cdot (r_t(\phi(i)) - r_t(i))].
\end{align*}
Interestingly, the term ``internal regret'' does not actually appear in~\cite{foster1998asymptotic}.

\cite{foster1999regret} use a stricter definition of internal regret that replaces the $\max_{\phi(i) \in [n]}$ with a $\sum_{\phi(i) \in [n]}\max\{0,\cdot\}$: 
\begin{align*}
\texttt{internal-regret} \text{~\cite{foster1999regret}}:= \E_{\text{algorithm's randomness}}[\sum_{i\in [n]}\sum_{\phi(i) \in [n]}\max\{0,\sum_{t\in [T]} \id_{\{p^t = i\}}\cdot (r_t(\phi(i)) - r_t(i))\}].
\end{align*}

Finally, most authors today use a more lenient definition of internal regret than our swap regret, that replaces $\sum_{i\in [n]}$ with $\max_{i\in [n]}$. 
\begin{align*}
\texttt{internal-regret} \text{~\cite{stoltz2005internal}}:= \max_{i\in [n]}\max_{\phi(i) \in [n]}\sum_{t\in [T]} \E_{i \sim p_t} [r_t(\phi(i)) - r_t(i)].
\end{align*}

Note that all of those notions are equivalent up to $\poly(n)$ factors, i.e.~if one approaches zero they all approach zero. However, in this work we focus on obtaining $\eps T$ regret for constant $\eps >0$, so these notions are not equivalent. In particular, using the common notion of internal regret, playing a uniformly random strategy trivially obtains $T/n$ regret. (A similar issue arises for $\eps$-approximate correlated equilibrium; see discussion in the introduction and also in~\cite{ganor2018communication}.)
	\section{Missing proof from Section \ref{sec:application}}
\label{sec:application-app}

We first provide the proof of Lemma \ref{lem:sparsification}.
\begin{proof}[Proof of Lemma \ref{lem:sparsification}]
Let $R_{i} = \sum_{j \in [n]}p_{i, j} \in [0, 1]$ be the marginal distribution of the $i$-th row and $C_{j} = \sum_{i\in [n]}p_{i, j}$ be the marginal distribution of the $j$-th row. Let $r_{i} = (p_{i, 1}, \ldots, p_{i, n}) \in [0,1]^{n}$ be the $i$-th row of $p$.

We sample $D = O(S^2\log(n)/\delta^2)$ rows $i_{1}, \ldots i_{D}$ from the distribution $\{R_{i}\}_{i \in [n]}$ and set
\begin{align}
p' = \frac{1}{D} \sum_{d=1}^{D} \frac{1}{R_{i_d}} e_{i_d} \otimes r_{i_{d}}. \label{eq:sparse1}
\end{align}
That is, $p'$ is obtained from $p$ by sampling $D$ rows and proper normalization.

It is clear that $p'$ has row support size at most $D = O(S^2\log(n)/\delta^2)$, and column support size at most $S$. We prove $p'$ is an $(\eps + \delta)$-approximate correlated equilibrium with high probability.
We first verify the row player. Let $\phi$ be the swap function that obtains the maximum utility under $p$. For $p'$, the distribution of each row either becomes $\mathbf{0}$ or gets scaled, and therefore, $\phi$ remains the optimal swap function. 
Hence, with probability $1 - 1/n^{\omega(1)}$, we have
\begin{align*}
&~ \E_{(a_1, a_2) \sim p'}[u_{1}(\phi(a_1); a_{2})] - \E_{(a_1, a_2) \sim p'}[u_{1}(a_1; a_{2})]\\
= &~ \frac{1}{D} \sum_{d=1}^{D}  \E_{a_2 \sim r_{i_d}/R_{i_d}} [u_{1}(\phi(i_{d}); a_{2}) -  u_{1}(i_d; a_{2})]\\
= &~ \E_{(a_1, a_2) \sim p}[u_{1}(\phi(a_1), a_2) - u_{1}(a_1; a_{2})] + O\left(\sqrt{\log(n)/D}\right)
\leq \eps + \delta.
\end{align*}
The first step follows the definition of $p'$ (see Eq.~\eqref{eq:sparse1}), the second step follows from Chernoff bound and the third step holds since $p$ is an $\eps$-approximate correlated equilibrium. 

We next verify the column player. Fix any column $j, j'\in [n]$, with probability $1 - 1/n^{\omega(1)}$, we have
\begin{align}
\E_{(a_1, a_2)\sim p',a_2 = j}[u_2(j'; a_1) ] = &~ \frac{1}{D} \sum_{d=1}^{D} \frac{p_{i_d, j}}{R_{i_d}}u_{2}(j'; i_{d})\notag \\
= &~ \sum_{i\in [n]}p_{i, j}u_{2}(j'; i) \pm O\left(\sqrt{\log(n)/D}\right)\notag \\
=&~ \E_{(a_1, a_2)\sim p, a_2 = j}[u_2(j'; a_1)] \pm \delta/2S.\label{eq:sparse2}
\end{align}
Here, the first step follows from the definition of $p'$ (see Eq.~\eqref{eq:sparse1}), the second step follows from Chernoff bound and the last follows from the choice of parameters.
We taking an union bound over all $j', j \in [n]$ and condition on this event in the rest of the proof. 

Let 
\[
C(p) := \{j: \exists i\in [n], p_{i, j}\neq 0\} \subseteq [n]
\]
be column supports of $p$. For any swap function $\phi$, we have
\begin{align*}
&~ \E_{(a_1, a_2)\sim p'}[u_2(\phi(a_2); a_1)] - \E_{(a_1, a_2)\sim p'}[u_2(a_2; a_1)] \\
= &~ \sum_{j \in C(p)}\E_{(a_1, a_2)\sim p', a_2 = j}[u_2(\phi(a_2); a_1) - u_2(a_2; a_1)]\\
\leq &~ \sum_{j\in C(p)}\E_{(a_1, a_2)\sim p, a_2 = j}[u_2(\phi(a_2); a_1) - u_2(a_2; a_1)] +  \delta/S \\
= &~ \E_{(a_1, a_2)\sim p}[u_2(\phi(a_2); a_1) - u_2(a_2; a_1)] + \delta  = \eps + \delta.
\end{align*}
The second step follows from Eq.~\eqref{eq:sparse2}, the third step holds since $|C(p)| = S$, the last step holds since $p$ is an $\eps$-approximate correlated equilibrium.
We complete the proof here.
\end{proof}

We next provide the proof of Lemma \ref{lem:efg-mwu} and give the efficient implementation of MWU for EFGs. The overall idea is simple and we sample the strategy according to a partition function, which can be recursively computed. The idea has been exploited for complete information game \cite{chhablani2023multiplicative,farina2022kernelized} and we provide a proof for incomplete information game (i.e., with chance player).

\paragraph{Additional notation} We fix strategies  $s_{-i, 1}, \ldots, s_{-i, T}$ of players $[m]\backslash\{i\}$ in the rest of proof. 
For any nodes $\nu \in \N_i$, let $\bar{u}_i(s_i; \nu)$ be the average utility of player $i$ if it visits node $\nu$, i.e., 
$\bar{u}_i(s_i; \nu) := \eta \sum_{t \in [T]} u_{i}(s_i, s_{-i, t}; \nu).$
We use $\bar{u}_i(s_i)$ to denote the utility at the root node and one has $p(s_i) \propto \exp(\bar{u}_i(s_i))$.
For any information set $h \in \mH_i$, define $\bar{u}_i(s_i; h)$ be the total utility of player $i$ if it visits the information set $h$, i.e.,
$\bar{u}_i(s_i; h) := \sum_{\nu \in h} \bar{u}_i(s_i; \nu).$

The information sets of player $i$ form a directed tree. Given an information set $h \in \mH_i$, let $\mT_{h} \subseteq \mH_i$ contain all information sets in the subtree rooted at $h$. Let $\mC_h$ contain all child information sets of $h$.
For any action $a \in A_{h}$, let $\mC_{h.a}$ contain all child information sets of $h$ that could be reached when player $i$ takes action $a$ at $h$. 
Slightly abuse of notation, we view all terminal nodes directly reachable from $h.a$ (i.e. not through other information set) as an information set of player $i$, and its action set contains only a dummy action $\emptyset$. 
The sets $\{\mC_{h.a}\}_{a \in A_h}$ form a partition of $\mC_{h}$, i.e., $\mC_h = \cup_{a \in A_h} \mC_{h.a}$.

We first make a few simple observations. 
\begin{lemma}
\label{lem:efg1}
For any strategy $s_i, s_i' \in \mS_i$ and any information set $h \in \mH_i$, if $s_i$ and $s_i'$ use the same actions for subtree $\mT_{h}$ and along the root path to $h$, then $\bar{u}_i(s_i; h) = \bar{u}_i(s_i'; h)$.
\end{lemma}

\begin{proof}
We have
\begin{align*}
\bar{u}_i(s_i; h) = &~ \sum_{\nu \in h}\bar{u}_{i}(s_i; \nu)  = \eta \sum_{\nu \in h}\sum_{t\in [T]}u_i(s_i, s_{-i, t};\nu) = \eta \sum_{\nu \in h}\sum_{t\in [T]}\sum_{z \in \mZ: \nu \preceq z} \pi(s_i, s_{-i, t}; z) \cdot \gamma_i(z)\\
=&~ \eta \sum_{\nu \in h}\sum_{t\in [T]}\sum_{z \in \mZ: \nu \preceq z} \pi(s_i', s_{-i, t}; z) \cdot \gamma_i(z) = \eta \sum_{\nu \in h}\sum_{t\in [T]}u_i(s_i', s_{-i, t};\nu) \\
= &~ \sum_{\nu \in h}\bar{u}_{i}(s_i'; \nu) = \bar{u}_i(s_i'; h).
\end{align*}
Here the first three steps are due to the definitions of $\bar{u}_i(s_i; h)$, $\bar{u}_i(s_i; \nu)$ and $u_{i}(s_i, s_{-i, t}; \nu)$.
The fourth step holds since for any terminal node $z$, one has $\pi(s_i, s_{-i, t}; z) = \pi(s_i', s_{-i, t}; z)$ if $s_i, s_i'$ use the same actions along the root path to $z$. The last three steps follow from the definition of $u_i(s_i',s_{-i,t}; h)$, $\bar{u}_i(s_i'; \nu)$, $\bar{u}_i(s_i'; h)$. We complete the proof here.
\end{proof}

\begin{lemma}
\label{lem:efg2}
For any strategy $s_i \in \mS_i$ and information set $h \in \mH_i$, we have 
\begin{align*}
\bar{u}_i(s_i, h) = \sum_{h' \in \mC_h}\bar{u}_i(s_i, h').
\end{align*}
\end{lemma}
\begin{proof}
For the first claim, we have
\begin{align*}
\bar{u}_i(s_i, h) =&~ \sum_{\nu \in h} \bar{u}_i(s_i, \nu) = \sum_{\nu \in h}\sum_{t\in [T]}\eta u_{i}(s_i, s_{-i, t}; \nu)\\
=&~ \sum_{\nu \in h}\sum_{t\in [T]}\sum_{z \in \mZ: \nu \preceq z}\eta \pi(s_i, s_{-i, t}; z) \cdot \gamma_i(z)\\
= &~\sum_{h' \in \mC_h}\sum_{\nu \in h'}\sum_{z \in \mZ: \nu \preceq z}\sum_{t\in [T]}\eta \pi(s_i, s_{-i, t}; z) \cdot \gamma_i(z)\\
= &~ \sum_{h' \in \mC_h}\sum_{\nu \in h'} \sum_{t\in [T]} \eta u_i(s_i, s_{-i,t}; \nu)\\
= &~ \sum_{h' \in \mC_h}\sum_{\nu\in h'}\bar{u}_i(s_i, h') =  \sum_{h' \in \mC_h}\bar{u}_i(s_i, h').
\end{align*}
Here the first three steps are due to the definitions of $\bar{u}_i(s_i, h)$, $\bar{u}_i(s_i, \nu)$ and $u_{i}(s_i, s_{-i, t}; \nu)$. The fourth step rearranges all terminal nodes in subtrees rooted at $h$. The last three steps are due to the definitions of $u_{i}(s_i, s_{-i, t}; \nu)$,  $\bar{u}_i(s_i, \nu)$ and  $\bar{u}_i(s_i, h')$.
\end{proof}

\paragraph{Equivalent class}
Given an information set $h \in \mH_i$, we write $s_{i} \stackrel{h}{\sim} s_{i}'$ if strategies $s_i$ and $s_i'$ use the same actions over information sets in the subtree $\mT_h$, and we say $s_i, s_i'$ are in the same equivalent class of $h$. 
Given an information set $h$, the strategy set $\mS_{i, h}\subseteq \mS_i$ takes exactly one strategy $s_i$ from each equivalent class of $h$, and $h$ is reachable under this strategy $s_i$ (i.e., $s_i$ uses the same the actions as $\sigma(h)$ along the root path to $h$).


We can now define the partition function over information sets.
\begin{definition}[Partition function]
The partition function $V_{i}:\mH_i\rightarrow \R$ is defined over each information set $h \in \mH_i$, such that
\begin{align*}
V_i(h) := \sum_{s_i\in \mS_{i, h}}\exp(\bar{u}_i(s_i; h)).
\end{align*}
The partition function $U_{i}: \mH_i \times A_i \rightarrow \R$ is defined over an information set $h\in \mH_i$ and action $a \in A_{h}$ pair, such that
\begin{align*}
U_i(h.a) := \sum_{s_i\in \mS_{i, h} \wedge s_i(h) = a}\exp(\bar{u}_i(s_i; h)).
\end{align*}
\end{definition}
As a simple corollary of Lemma \ref{lem:efg1}, the value of partition function $V_i(h)$ and $U_i(h.a)$ does not depend on the exact choice of strategy from each equivalent class.

\paragraph{Utility of terminal nodes} When the information set $h$ is made up of terminal nodes, i.e., $h \subseteq \mZ$, then $\mS_{i, h}$ contains only one strategy $s_{i}$ according to our definition. Let $\Lambda_{i}(h) = \bar{u}_i(s_i; h) = \eta\sum_{t\in [T]}u_i(s_i, s_{-i, t}; h)$ be the utility of player $i$ at $h$, for any $h \subseteq \mZ$.  
It is common in the literature (e.g. \cite{zinkevich2007regret}) to assume the utility of information sets can be computed efficiently for any fixed strategy $s \in \mS$, so does the value $\Lambda_{i}(h)$. In the case that the exact utility of an information set can not be computed efficiently, one can draw $O(\log(|\mS|)/\eps^2) = O(m\Phi \log(n)/\eps^2)$ EFGs (without chance nodes) from the prior distribution of chance nodes and compute the average utility on these EFGs, it gives good approximation for all strategies $\mS$. We omit the details here.

The partition function can be computed efficiently via recursion.
\begin{lemma}[Recursive computation of partition function]
\label{lem:partition-recursive}
The partition function can be computed recursively
\begin{align*}
V_{i}(h) = 
\left\{
\begin{matrix}
\exp(\Lambda_i(h)) & h \subseteq \mZ\\
\sum_{a \in A_h} U_i(h.a) & \text{otherwise}
\end{matrix}
\right.
\end{align*}
and 
\begin{align*}
U_{i}(h.a) = \prod_{h' \in \mC_{h.a}}V_i(h') \cdot \prod_{h'' \in \mT_h\setminus (\mC_{h.a} \cup \{h\})}|A_{h''}|.
\end{align*}
\end{lemma}
\begin{proof}
For the first claim, if $h$ is made up of terminal nodes, i.e., $h\subseteq \mZ$, then there is only one strategy $s_i$ in $\mS_{i, h}$, and we have $V_i(h) = \exp(\bar{u}_i(s_i; h)) = \exp(\Lambda_i(h))$.
On the other hand, if $h$ is made up of decision nodes, then we have
\begin{align*}
V_i(h) =  \sum_{s_i\in \mS_{i, h}}\exp(\bar{u}_i(s_i, h)) = \sum_{a \in A_h}\sum_{s_i\in \mS_{i, h} \wedge s_i(h) = a}\exp(\bar{u}_i(s_i, h)) = \sum_{a \in A_h}U_i(h.a).
\end{align*}
For the second claim, we have
\begin{align}
U_{i}(h.a) = \sum_{s_i\in \mS_{i, h}\wedge s_i(h) = a}\exp(\bar{u}_i(s_i, h))
=  \sum_{s_i\in \mS_{i, h}\wedge s_i(h) = a}\exp\left(\sum_{h' \in \mC_h}\bar{u}_i(s_i, h')\right).  \label{eq:efg-mwu1}
\end{align}
The first step follows from the definition of partition function $U_{i}(h.a)$, the second step follows from Lemma \ref{lem:efg2}.

For the RHS of Eq.~\eqref{eq:efg-mwu1}, we have
\begin{align}
\sum_{s_i\in S_{i, h}\wedge s_i(h) = a}\exp\left(\sum_{h' \in \mC_h}\bar{u}_i(s_i, h')\right) = &~  \sum_{s_i\in S_{i, h}\wedge s_i(h) = a}\exp\left(\sum_{h'\in \mC_{h.a}}\bar{u}_i(s_i; h')\right)\notag \\
= &~ \sum_{s_i\in S_{i, h}\wedge s_i(h) = a} \prod_{h'\in \mC_{h.a}} \exp(\bar{u}_i(s_i; h'))\notag \\
= &~ \prod_{h'\in \mC_{h.a}} \left(\sum_{s_i \in S_{i, h'}}\exp(\bar{u}_i(s_i; h') \right) \cdot \prod_{h'' \in \mT_h\setminus (\mC_{h.a} \cup \{h\})}|A_{h''}| \notag \\
= &~\prod_{h'\in \mC_{h.a}} V_i(h') \cdot \prod_{h'' \in \mT_h\setminus (\mC_{h.a} \cup \{h\})}|A_{h''}| \label{eq:efg-mwu2} 
\end{align}
The first step holds since for any information set $h' \in \mC_h\setminus \mC_{h.a}$, one has $\bar{u}_i(s_i; h') = 0$. This is because the player $i$ never visits $h'$ given its strategy $s_i$ satisfies $s_i(h) = a$.
%
In the third step, we exchange the product and summation, this is valid due to Lemma \ref{lem:efg1}.
The last step follows from the definition of $V_i(h')$.

Combining Eq.~\eqref{eq:efg-mwu1}\eqref{eq:efg-mwu2}, we complete the proof.
\end{proof}

Lemma \ref{lem:partition-recursive} gives a way of computing the partition function. We next show how to sample from the distribution in Eq.~\eqref{eq:mwu-distribution} using partition functions. 
It is wlog to assume the root of $\mT$ is a decision node of player $i$. Consider the directed tree formed by information sets $\mH_i$, the sampling process assigns actions to information sets in a top-down fashion, from the root to leaves. In particular, consider an arbitrary ordering of information sets $h_1, \ldots, h_{\Phi} \in \mH_i$, such that information sets at higher level come earlier than information sets at lower level, then we have
\begin{lemma}[Sampling with partition function]
\label{lem:efg-sample}
Suppose the distribution $p \in \Delta(\mS_i)$ is given as Eq.~\eqref{eq:mwu-distribution}, then one can sample a strategy $s_i$ from $p$ as follow:
For $t = 1,2,\ldots, \Phi$
\begin{align}
\Pr[s_i(h_t) = a_t]  = \left\{ \begin{matrix}
\frac{U_i(h_t.a_t)}{V_i(h_t)} & h_t \text{ is reachable under } s_i  \label{eq:sample-rule}\\
\frac{1}{|A_{h_t}|} & \text{otherwise}
\end{matrix}
\right. \quad \quad \forall a_t \in A_{h_t}.
\end{align}
\end{lemma}
\begin{proof}
For any $t \in [\Phi]$, and for any action $a_{1} \in A_{h_1}, \ldots, a_{t} \in A_{h_t}$, we prove
\begin{align}
\Pr[s_i(h_{1}) = a_1, \ldots,  s_i(h_{t}) = a_t] = \frac{\sum_{s \in \mS_i \wedge \{s(h_\tau) = a_\tau\}_{\tau \in [t]}}\exp(\bar{u}_i(s))}{\sum_{s \in \mS_i}\exp(\bar{u}_i(s))} \label{eq:sample-goal}
\end{align}
by induction. 

The base case of $t= 0$ holds trivially. Suppose the claim continues to hold up to $t$, then for $t+1$, for any action $a_{1} \in A_{h_1} \ldots, a_{t+1} \in A_{h_{t+1}}$, by the inductive hypothesis, we have
\begin{align}
&~\Pr[s_i(h_{1}) = a_1,\ldots, s_i(h_{t+1}) = a_{t+1}] \notag \\
= &~ \Pr[s_{i}(h_1) = a_{1}, \ldots, s_i(h_t) = a_t]  \cdot \Pr[s_i(h_{t+1}) = a_{t+1} |s_{i}(h_1) = a_{1}, \ldots, s_i(h_t) = a_t] \notag \\
= &~  \frac{\sum_{s \in \mS_i \wedge \{s(h_\tau) = a_\tau\}_{\tau \in [t]}}\exp(\bar{u}_i(s))}{\sum_{s \in \mS_i}\exp(\bar{u}_i(s))} \cdot \Pr[s_i(h_{t+1}) = a_{t+1} |s_{i}(h_1) = a_{1}, \ldots, s_i(h_t) = a_t].
\label{eq:sample-3}
\end{align}

We divide into two cases. 

{\bf Case 1.} Suppose the information set $h_{t+1}$ is not reachable from $s_i$, given $s_i(h_1) = a_1, \ldots, s_i(h_t) = a_t$. Then due to the sampling rule (Eq.~\eqref{eq:sample-rule}) we have
\begin{align}
\label{eq:sample-1}
\Pr[s_i(h_{t+1}) = a_{t+1} |s_{i}(h_1) = a_{1}, \ldots, s_i(h_t) = a_t] = \frac{1}{|A_{h_{t+1}}|}.
\end{align}
Moreover, the choice of $s(h_{t+1}) \in A_{h_{t+1}}$ does not affect the total utility $\bar{u}_i(s)$ given $h_{t+1}$ is not reachable from $s$, then we have
\begin{align}
\frac{\sum_{s \in \mS_i \wedge \{s(h_{\tau}) = a_{\tau}\}_{\tau\in [t+1]}}\exp(\bar{u}_i(s))}{\sum_{s \in \mS_i \wedge \{s(h_{\tau}) = a_{\tau}\}_{\tau\in [t]}}\exp(\bar{u}_i(s))} = \frac{1}{|A_{h_{t+1}}|}. \label{eq:sample-2}
\end{align}
Combining Eq.~\eqref{eq:sample-3}\eqref{eq:sample-1}\eqref{eq:sample-2}, we have proved Eq.~\eqref{eq:sample-goal} for the first case.

{\bf Case 2.} Suppose the information set $h_{t+1}$ is reachable from $s_i$, given $s_i(h_1) = a_1, \ldots, s_i(h_t) = a_t$. 
Then, according to the sampling rule (Eq.~\eqref{eq:sample-rule}), we have
\begin{align}
\Pr[s_i(h_{t+1}) = a_{t+1} |s_{i}(h_1) = a_{1},\ldots, s_{i}(h_t) = a_{t}] = &~ \frac{U_i(h_{t+1}.a_{t+1})}{V_i(h_{t+1})}\notag \\
= &~ \frac{\sum_{s\in \mS_{i, h_{t+1}} \wedge s_i(h_{t+1}) = a_{t+1}}\exp(\bar{u}_i(s; h_{t+1}))}{\sum_{s\in \mS_{i, h_{t+1}} }\exp(\bar{u}_i(s; h_{t+1}))}.
\label{eq:sample-4} 
\end{align}
The second step follows from the definition of partition functions.

Suppose the information set $h_{t+1}$ is at level $\ell$ and let $\mH_{i, \ell}\subseteq \mH_{i}$ contain all information sets at level $\ell$. Let $R_{t} \subseteq \mH_{i,\ell}$ be all information sets that are reachable from strategy $s_i$, given $s_i(h_1) = a_1, \ldots, s_i(h_t) = a_t$. 
For any $h \in \mH_i$, define
\begin{align*}
\mS_{i, h \mid \{h_{\tau},a_{\tau}\}_{\tau \in [t]}} =\left\{ 
\begin{matrix}
\emptyset & h \text{ is not reachable given } \{s(h_{\tau}) = a_{\tau}\}_{\tau \in [t]}\\
\mS_{i, h} & h \text{ is reachable given } \{s(h_{\tau}) = a_{\tau}\}_{\tau \in [t]} \text{ and } h\notin \{h_\tau\}_{\tau \in [t]}\\
\{s \in \mS_{i, h}: s(h) = a_\tau\} &~ h \text{ is reachable given } \{s(h_{\tau}) = a_{\tau}\}_{\tau \in [t]} \text{ and } h= h_\tau
\end{matrix}
\right..
\end{align*}
Then we have
\begin{align}
\frac{\sum_{s \in \mS_i \wedge \{s(h_{\tau}) = a_{\tau}\}_{\tau\in [t+1]}}\exp(\bar{u}_i(s))}{\sum_{s \in \mS_i \wedge \{s(h_{\tau}) = a_{\tau}\}_{\tau\in [t]}}\exp(\bar{u}_i(s))} 
= &~ \frac{\sum_{s \in \mS_i \wedge \{s(h_{\tau}) = a_{\tau}\}_{\tau\in [t+1]}}\exp\left(\sum_{h \in \mH_{i, \ell}}\bar{u}_i(s; h)\right)}{\sum_{s \in \mS_i \wedge \{s(h_{\tau}) = a_{\tau}\}_{\tau\in [t]}}\exp\left(\sum_{h\in \mH_{i,\ell}}\bar{u}_i(s; h)\right)} \notag \\
= &~ \frac{\prod_{h \in \mR_{t}}\left(\sum_{s \in \mS_{i, h \mid \{h_{\tau}, a_{\tau}\}_{\tau\in [t+1]}}}\exp(\bar{u}_i(s; h))\right)  }{\prod_{h \in  \mR_{t}} \left(\sum_{s \in \mS_{i, h \mid \{h_{\tau}, a_{\tau}\}_{\tau \in [t]}}}\exp(\bar{u}_i(s; h))\right)}\notag \\
= &~ \frac{\sum_{s \in \mS_{i, h_{t+1} \mid \{h_{\tau}, a_{\tau}\}_{\tau\in [t+1]}}}\exp(\bar{u}_i(s; h_{t+1}))}{\sum_{s \in \mS_{i, h_{t+1} \mid \{h_{\tau}, a_{\tau}\}_{\tau \in [t]}}}\exp(\bar{u}_i(s; h_{t+1}))}\notag \\
= &~ \frac{\sum_{s \in \mS_{i, h_{t+1}} \wedge s(h_{\tau+1}) = a_{\tau+1}}\exp(\bar{u}_i(s; h_{t+1}))}{\sum_{s \in \mS_{i, h_{t+1}}}\exp(\bar{u}_i(s; h_{t+1}))}. \label{eq:sample-5}
\end{align}
The first step follows from repeatedly applying Lemma \ref{lem:efg2} to information sets at level $1,2,\ldots, \ell-1$. In the second step, we exchange the product and summation, this is valid due to Lemma \ref{lem:efg1}. The third and the fourth step follow from the definition of $\mS_{i, h\mid \{h_\tau, a_{\tau}\}_{\tau \in [t]}}$

Combining Eq.~\eqref{eq:sample-3}\eqref{eq:sample-4}\eqref{eq:sample-5}, we have proved \eqref{eq:sample-goal} for the second case.

We have finished the induction. The correctness of sampling procedure follows directly by plugging $t = \Phi$ to Eq.~\eqref{eq:sample-goal}. We complete the proof here.
\end{proof}

Combing Lemma \ref{lem:partition-recursive} and Lemma \ref{lem:efg-sample}, we complete the proof for Lemma \ref{lem:efg-mwu}.
	\section{Missing proof from Section \ref{sec:lower}}
\label{sec:lower-app}

We first present the missing details of the technical Lemma \ref{lem:tech}.
Let 
\[
R_{i^{*}} = \sum_{h\in [H]}(p_h(1) + p_h(2)) r_{h}(i^*) \quad \text{and} \quad 
R_{\ALG} = \sum_{h\in [H]}\sum_{i \in [2]}p_{h}(i)r_{h}(i)
\]
be the total (weighted) reward of $i^*$ and the total reward of the algorithm. 
Let $\D_1, \D_2$ be the distribution of two coins. For any $h\in [H]$, let $r_{1:h} = (r_1, \ldots, r_h) \sim (\D_1\times\D_2)^h$ be the reward of the first $h$ days.

First, the reward of $i^*$ satisfies
\begin{lemma}
\label{lem:tech-opt}
    We have 
    \[
    \E[R_{i^{*}}] = \left(\frac{1}{2}+\Delta\right) \cdot \sum_{h\in [H]}\E[p_h(1) + p_h(2)].
    \]
\end{lemma}
\begin{proof}
Telescoping over $h \in [H]$, we have
\begin{align*}
\E[R_{i^{*}}] = &~  \sum_{h\in [H]} \E\left[(p_h(1) + p_h(2))r_h(i^{*})\right]\\
= &~ \sum_{h\in [H]} \E[r_h(i^{*})]\cdot \E[p_h(1) + p_h(2)]\\
= &~ (\frac{1}{2} + \Delta) \sum_{h\in [H]}\E[p_h(1) + p_h(2)].
\end{align*} 
The second step follows from $r_{h}(i^{*})$ is independent of $p_h$ and the third step follows from $\E[r_h(i^{*})] = \frac{1}{2} + \Delta$.
\end{proof}

The following bound on the Bernoulli distribution $B_{1/2}^h$ and $B_{1/2+\Delta}^h$ is standard. 
\begin{lemma}
\label{lem:tv}
Let $\Delta \in (0,1/20]$ and $H = \frac{1}{400\Delta^2}$.
For any $h\in [H]$, we have 
\[
\TV(B_{1/2}^h, B_{1/2+\Delta}^{h}) \leq \frac{1}{10}.
\]
\end{lemma}
\begin{proof}
For any $h\in [H]$, we have
\begin{align*}
\TV(B_{1/2}^h, B_{1/2+\Delta}^{h}) \leq &~ \sqrt{\frac{1}{2}\KL(B_{1/2}^h || B_{1/2+\Delta}^{h})}\\
= &~ \sqrt{\frac{h}{2} \KL(B_{1/2}|| B_{1/2+\Delta})}\\
\leq &~ 2\sqrt{h}\Delta \leq \frac{1}{10}.
\end{align*}
The first step follows from Pinsker inequality, the second step follows the independence, the third step follows from $\KL(B_{1/2} || B_{1/2 +\Delta}) \leq 8\Delta^2$ and the last step follows from $H = 1/400\Delta^2$.
\end{proof}

Next, we bound the reward of algorithm.
\begin{lemma}
\label{lem:tech-algo}
For any algorithm, we have 
\[
\E[R_{\ALG}] \leq \left(\frac{1}{2}+\frac{\Delta}{2}\right)\sum_{h \in [H]}\E[p_h(1) + p_{h}(2)] + \frac{3}{20}\Delta H.
\]
\end{lemma}
\begin{proof}
First, we telescope $R_{\ALG}$ over $h\in [H]$
\begin{align}
\E[R_{\ALG}]  = &~ \sum_{h\in [H]}\E\left[\sum_{i \in [2]}p_{h}(i)r_{h}(i)\right]\notag \\
= &~ \sum_{h\in [H]}\E\left[\sum_{i \in [2]}p_{h}(i)(r_{h}(i) -1/2)\right] + \frac{1}{2}\sum_{h\in [H]}\E[p_h(1) + p_h(2)]
\label{eq:tech2}
\end{align}
where the first step follows from the linearity of the expectation.

We bound the RHS of Eq.~\eqref{eq:tech2}. For any fixed $h\in [H]$, we have 
\begin{align}
&~ \E\left[\sum_{i \in [2]}p_{h}(i)(r_{h}(i) -1/2)\right] \notag \\
= &~ \frac{1}{2} \E\left[\sum_{i \in [2]}p_{h}(i)(r_{h}(i) -1/2) | i^{*} = 1\right] + \frac{1}{2} \E\left[\sum_{i \in [2]}p_{h}(i)(r_{h}(i) -1/2) | i^{*} = 2\right]\notag \\
= &~ \frac{1}{2}\E_{r_{1:h-1} \sim (B_{1/2 + \Delta} \times B_{1/2})^{h-1}}\left[\E_{r_h \sim B_{1/2+\Delta} \times B_{1/2}}\left[\sum_{i\in [2]}p_h(i)(r_h(i) - 1/2) | r_{1:h-1}\right]\right] \notag\\
&~ +\frac{1}{2}\E_{r_{1:h-1} \sim (B_{1/2} \times B_{1/2 + \Delta})^{h-1}}\left[\E_{r_h \sim B_{1/2} \times B_{1/2 + \Delta}}\left[\sum_{i\in [2]}p_h(i)(r_h(i) - 1/2) | r_{1:h-1}\right]\right] \notag\\
\leq &~ \frac{1}{2}\E_{r_{1:h-1} \sim (B_{1/2} \times B_{1/2})^{h-1}}\left[\E_{r_h \sim B_{1/2+\Delta} \times B_{1/2}}\left[\sum_{i\in [2]}p_h(i)(r_h(i) - 1/2) | r_{1:h-1}\right]\right] \notag\\
&~ +\frac{1}{2}\E_{r_{1:h-1} \sim (B_{1/2} \times B_{1/2})^{h-1}}\left[\E_{r_h \sim B_{1/2} \times B_{1/2+\Delta}}\left[\sum_{i\in [2]}p_h(i)(r_h(i) - 1/2) | r_{1:h-1}\right]\right] + \frac{1}{10}\Delta\notag\\
= &~ \frac{1}{2}\Delta \cdot \E_{r_{1:h-1}\sim (B_{1/2}\times B_{1/2})^{h-1}}\left[\E[p_{h}(1) + p_h(2) | r_{1:h-1}]\right]  + \frac{1}{10}\Delta\notag\\
\leq &~ \frac{1}{2}\Delta\cdot  \E_{r_{1:h-1}\sim (\D_1\times\D_2)^{h-1}}\left[\E[p_h(1) + p_h(2) | r_{1:h-1}]\right]   + \frac{1}{20}\Delta + \frac{1}{10}\Delta\notag\\
= &~ \frac{1}{2}\Delta\cdot  \E[p_h(1) + p_h(2)]  +\frac{3}{20}\Delta \label{eq:tech1}.
\end{align}
The first step holds since $i^{*}$ is chosen uniformly at random from $[2]$. The second step follows from the law of expectation and the fact that $p_h$ is determined by $r_{1:h-1}$.
The third step holds since (1) $\TV(B_{1/2}^{h-1} , B_{1/2+\Delta}^{h-1}) \leq \frac{1}{10}$ (see Lemma \ref{lem:tv}) and (2) for fixed $r_{1:h-1}$,
\begin{align*}
\E_{r_h\sim B_{1/2+\Delta}\times B_{1/2}}\left[\sum_{i\in [2]}p_h(i)(r_h(i) - 1/2) | r_{1:h-1}\right]   \leq \Delta 
\end{align*}
and
\begin{align*}
\E_{r_h\sim B_{1/2}\times B_{1/2+\Delta}}\left[\sum_{i\in [2]}p_h(i)(r_h(i) - 1/2) | r_{1:h-1}\right] \leq \Delta.
\end{align*}
The fourth step holds since $p_h$ is determined by $r_{1:h-1}$. The fifth step holds since $\TV((B_{1/2}\times B_{1/2})^{h-1}, (\D_1\times \D_2)^{h-1}) \leq \frac{1}{10}$ (see Lemma \ref{lem:tv}) and $p_h(1) + p_h(2) \leq 1$ and the last step follows from the law of expectation.

Combining Eq.~\eqref{eq:tech1} and Eq.~\eqref{eq:tech2}, we have
\begin{align*}
\E[R_{\ALG}] = \left(\frac{1}{2} + \frac{1}{2}\Delta\right) \sum_{h\in [H]}\E[p_h(1) + p_h(2)] + \frac{3}{20}\Delta H.
\end{align*}
We complete the proof here.
\end{proof}

Combining Lemma \ref{lem:tech-opt} and \ref{lem:tech-algo}, we conclude the proof of Lemma \ref{lem:tech}.

We next prove Lemma \ref{lem:alternative}
\begin{proof}[Proof of Lemma \ref{lem:alternative}]
We first analyse the LHS of Lemma \ref{lem:alternative}. By the definition of $X_i$, one has
\begin{align}
\sum_{i \in [n]}X_i = &~ \sum_{i \in [n]}\sum_{t \in [S_{a(i)} -1] }p_t(i) = \sum_{a \in [0:n/2-1]} \sum_{t \in[S_{a} - 1]}p_t(2a+1) + p_{t}(2a+2)\notag \\
= &~ \sum_{a \in [0:n/2-1]} \sum_{t \in [S_a: E_a]} \sum_{a' > a}p_t{(2a'+1)} + p_t(2a' + 2).\label{eq:alternative1}
\end{align}

For any node $a, b \in \mT$, we write $b \subseteq a$ if $b$ is a node in the subtree of $a$. 
For the RHS, we have
\begin{align}
\sum_{\ell \in [0:L-1]}\sum_{a \in \mT_{\ell}}M_a = &~ \sum_{\ell \in [0:L-1]}\sum_{a \in \mT_{\ell}} \sum_{t \in [S_a: E_a]}\sum_{i \in \N^{+}(a)}p_t(i) \notag\\
= &~ \sum_{\ell \in [0:L-1]}\sum_{a \in \mT_{\ell}} \sum_{b \in [0:n/2-1], b \subseteq a} \sum_{t\in [S_b: E_b]}\sum_{i\in \N^{+}(a)} p_t(i)\notag \\
= &~ \sum_{b \in [n/2-1]}\sum_{t\in [S_b: E_b]}\sum_{\ell \in [0:L-1]}\sum_{a \in \mT_\ell, b\subseteq a}\sum_{i \in \N^{+} (a)}p_t(i)\notag\\
= &~ \sum_{b \in [n/2-1]}\sum_{t\in [S_b: E_b]}\sum_{\ell\in [0:L-1]}\sum_{i \in \N^{+}(b_{L}\ldots b_{\ell + 1})}p_{t}(i)\notag \\
= &~ \sum_{b \in [n/2-1]}\sum_{t\in [S_b: E_b]} \sum_{b' > b}p_t(2b+1) + p_t(2b+2).\label{eq:alternative2}
\end{align}
The first step follows from the definition of $M_a$. In the second step, we split the interval $[S_a: E_a]$ of node $a$ into intervals of its leaf nodes $\cup_{b\in [0:n/2-1], b\subseteq a} [S_b: E_b]$. We exchange summation in the third step and the last step follows from the definition of $\N^{+}(b)$.

Combining Eq.~\eqref{eq:alternative1}\eqref{eq:alternative2}, we complete the proof.
\end{proof}

Finally, we prove Lemma \ref{lem:basic-expecation}
\begin{proof}[Proof of Lemma \ref{lem:basic-expecation}]
For any node $a \in \mT$, let $T_a$ be the number of days spent over node $a$.
For any level $\ell \in [0:L]$ and node $a \in \mT_\ell$, we prove 
\begin{align*}
\E[T_a | a \in \mV_{\ell}] = H C_K^{\ell} \quad \text{and} \quad E[T_a | a \notin \mV_{\ell}] = 0.
\end{align*}
We prove by induction on $\ell$. The claim holds trivially for $\ell = 0$ because Nature spends $H$ days over any leaf node $a$ it visits. Suppose it holds up to level $\ell-1$, then at level $\ell$, suppose Nature visits the node $a \in \mT_{\ell}$, then we have
\begin{align*}
\E[T_{a}| a\in \mV_{\ell}] = &~ \sum_{k=0}^{K-1}\E[T_{a.k} | a \in \mV_{\ell}] = \sum_{k=0}^{K-1}\left(1-\frac{1}{2K}\right)^{k}\E[T_{a, k} | a.k\in \mV_{\ell-1}]\\
= &~ H C_{K}^{\ell -1} \cdot C_{K} = H C_{K}^{\ell}.
\end{align*}
The second step holds since Nature skips each child node with probability $\frac{1}{2K}$, the third step follows from the inductive hypothesis.

Since Nature always visits the root node, one has
\begin{align*}
\E[T_{\ALG}] = H C_K^{L} \in \left[2^{-L} \cdot \frac{K^{L}}{400\Delta^{2}}, \frac{K^{L}}{400\Delta^{2}} \right].
\end{align*}
The last step follows from $C_K \in (K/2, K)$ and $H = 1/400\Delta^{2}$. This completes the proof.
\end{proof}

\end{document}